\documentclass[11pt]{article}

\usepackage[utf8]{inputenc}
\usepackage[a4paper,portrait,margin=1in]{geometry}
\usepackage{amsmath,amssymb,amsthm,mathtools,mathrsfs}
\usepackage[hidelinks]{hyperref}
\usepackage[dvipsnames]{xcolor}
\usepackage{graphicx}
\hypersetup{
    colorlinks=true,
    linkcolor=blue,
    citecolor=teal
}
\usepackage{microtype}
\usepackage{enumitem}
\usepackage{booktabs,tabularx}
\usepackage{mathpazo}
\usepackage{authblk}

\newtheorem{theorem}{Theorem}[section]
\newtheorem{lemma}[theorem]{Lemma}
\newtheorem{proposition}[theorem]{Proposition}
\newtheorem{corollary}[theorem]{Corollary}
\newtheorem{conjecture}[theorem]{Conjecture}
\theoremstyle{definition}
\newtheorem{definition}[theorem]{Definition}
\newtheorem{remark}[theorem]{Remark}

\newcommand{\R}{\mathbb R}
\newcommand{\C}{\mathbb C}
\newcommand{\eps}{\varepsilon}
\newcommand{\1}{\mathbf 1}
\newcommand{\ket}[1]{\lvert #1\rangle}

\newcommand{\norm}[1]{\left\lVert #1\right\rVert}
\newcommand{\abs}[1]{\left\lvert #1\right\rvert}
\newcommand{\E}{\mathbb E}
\newcommand{\Prb}{\mathbb P}
\newcommand{\dist}{\operatorname{dist}}

\DeclareMathOperator{\rank}{rank}
\DeclareMathOperator{\Tr}{Tr}
\DeclareMathOperator{\col}{col}

\newcount\Comments
\newcommand{\kibitz}[2]{\ifnum\Comments=1\textcolor{#1}{#2}\fi}

\title{Near-Optimal Quantum Lower Bounds for\\
Convex Optimization via Fourier Rank}
\author[1]{Brandon Augustino}
\author[1]{Shouvanik Chakrabarti\thanks{Authors are listed alphabetically. Corresponding author: \texttt{shouvanik.chakrabarti@jpmchase.com}.}}
\author[1]{Enrico Fontana}
\author[1]{Dylan Herman}
\author[1]{Junhyung Lyle Kim}
\author[1]{Guneykan Ozgul}
\author[1,2]{Nadezhda Voronova\footnote{Contributions made during the author's internship at JPMorganChase.}}
\affil[1]{Global Technology Applied Research, JPMorganChase, New York, NY 10001, USA}
\affil[2]{CNRS, Universit\'{e} Paris Cit\'{e}, IRIF, F-75013 Paris, France}
\date{}

\begin{document}
\maketitle

\begin{abstract}
We establish a linear quantum query lower bound for high-accuracy convex optimization over an explicit family of $n$-dimensional ellipsoids. We focus on linear optimization with an explicitly given objective, where the feasible set is accessed through a membership oracle. We show that any algorithm that, for every unit linear objective, returns an exactly feasible point with additive objective error $\Theta(n^{-1/2})$ requires $\Omega(n)$ membership queries, where the ellipsoids are contained in the unit ball. The same lower bound can be shown to hold if the returned point is only required to be approximately feasible, within $\Theta(n^{-3/2})$ distance from the feasible set. This resolves an open question posed by Chakrabarti, Childs, Li, and Wu~(\textit{Quantum}, 2020) and by van Apeldoorn, Gily\'en, Gribling, and de Wolf~(\textit{Quantum}, 2020). Coupled with the upper bounds in these papers, the query complexity of high-accuracy convex optimization is characterized tightly up to logarithmic factors. The proof is built around a lower bound for determinant computation that is derived via a novel polynomial method based on Fourier-rank. In the continuous matrix phase-query model, computing the determinant of a real $n\times n$ matrix requires at least $n/2$ matrix-vector product queries. The construction also yields an $\Omega(n)$ phase-query lower bound for estimating the minimum eigenvalue of a real symmetric $n\times n$ matrix to additive accuracy $\Theta(n^{-2})$. These results extend the determinant and minimum-eigenvalue lower bounds of Childs, Hung, and Li~(ICALP 2021) from finite fields to the real-valued setting. Based on the same constructions, we also prove an $\Omega(\min\{\sqrt\kappa,n\})$ gradient-query lower bound for constant-accuracy optimization of $1$-strongly convex and $\kappa$-smooth functions, which is optimal up to a logarithmic factor.
\end{abstract}

\subsection*{Changes from Version 1}
\begingroup
\small
Version 1 of this paper~\cite{augustino2026nearoptimalquantumlowerbounds} appeared on arXiv on September 8, 2026, the same day as concurrent and independent work by Childs~\cite{childs2026quantumlowerboundsconvex} that obtains similar results through closely related techniques. Version 2, submitted on October 5, 2026, contains the following updates:
\begin{enumerate}
    \item Version 1 proved an $\Omega\left(\frac{n}{\log n \log \log n}\right)$ lower bound and discussed a possible improvement to $\Omega\left(\frac{n}{\log n}\right)$. Version 2 strengthens the lower bound to $\Omega(n)$ through a reduction that uses the optimal value of one random linear objective. The same reduction strengthens the gradient-query lower bound for strongly convex quadratics to $\Omega(\min\{\sqrt\kappa,n\})$ at accuracy $1/200$. The inverse-iteration argument of Version 1 is retained in Section~\ref{sec:inverse-iteration-argument}.
    \item We add a detailed comparison with~\cite{childs2026quantumlowerboundsconvex}. This comparison was not possible in Version 1 because that paper was not yet publicly available.
    \item We clarify the precision regime in which our lower bounds hold, and give a more concrete discussion of open questions.
\end{enumerate}
\endgroup

\section{Introduction}

\subsection{Motivation}

Convex optimization is a cornerstone of mathematical programming, with applications spanning machine learning, statistics, signal processing, control, finance, and economics~\cite{boyd2004convex}. Beyond its direct relevance, methods for solving nonconvex optimization problems, such as integer programs, often rely on convex relaxations. Understanding the fundamental computational limits of convex optimization is therefore important in both theory and practice. The possibility of quantum speedups in optimization has recently been the subject of significant interest. Determining the extent of possible quantum speedup in convex optimization is central to the understanding of quantum advantage in optimization as a whole.

Let $K\subset\R^n$ be a convex body satisfying
\[
 B_2(0,r)\subseteq K\subseteq B_2(0,R),
\]
where $B_2(x,t)$ is the Euclidean ball of radius $t$ centered at $x$. Given a convex function $f:\R^n\to\R$ and an error tolerance $\eps>0$, the canonical problem is
\begin{equation}
 \text{find }x\in K\text{ such that }f(x)\le\min_{z\in K}f(z)+\eps.
 \label{eq:canonical-convex-optimization}
\end{equation}
Algorithms for this problem are studied in two general frameworks. In the \emph{white-box model}, the data defining $K$ and $f$ are stored explicitly, and complexity is measured in arithmetic operations. Most work on quantum speedups in this setting has focused on linear and semidefinite programming~\cite{van2018improvements,van2019games,van2020quantum,apers2026quantum,augustino2021quantum,augustino2023quantum,bouland2023quantum,brandao2017quantum,brandao2017exponential,brandao2019faster}; see~\cite{abbas2024challenges} for a survey.

In the \emph{black-box model}, the function $f$ and the set $K$ are accessed only through oracles, and query complexity is the number of oracle calls needed to solve \eqref{eq:canonical-convex-optimization}. This model, introduced by Nemirovskii and Yudin~\cite{nemirovskii1983problem}, captures problems in which the objective or feasible set is complicated or implicit. There are two main regimes within this model.

In the \emph{high-dimensional} regime, where $n$ is much larger than the geometric and accuracy parameters, one seeks methods whose query complexity is independent of $n$. Quantum gradient estimation~\cite{jordan2005fast,gilyen2019gradient,van2020convex,chakrabarti2020quantum} is the main source of quantum speedups in this regime. It has led to faster algorithms for both smooth and nonsmooth problems~\cite{kim2025fast}.

The complementary \emph{high-accuracy} regime seeks bounds that are polynomial in $n$ and logarithmic in $R/(r\eps)$. We focus on this regime. The objective is a known linear function, and the feasible set is accessed via a membership oracle. On input $x\in\R^n$, the oracle returns
\begin{equation}
 \mathsf{MEM}_K(x)=
 \begin{cases}
  1,&x\in K,\\
  0,&x\notin K.
 \end{cases}
 \label{eq:intro-membership}
\end{equation}
The algorithm may query this map in superposition over a set of represented points fixed in advance. Section~\ref{sec:prelim} gives the precise quantum definitions.

Two related oracles will help us place our results in context. A separation oracle $\mathsf{SEP}_K$ either certifies that a point belongs to $K$ or returns a hyperplane separating it from $K$. An optimization oracle $\mathsf{OPT}_K$ takes a linear objective and returns an approximately optimal point in $K$. Membership, separation, and optimization oracles are polynomially equivalent in the classical theory~\cite{grotschel2012geometric}. We use $\widetilde O$ to suppress factors logarithmic in the dimension, geometric parameters, and inverse accuracy. Lee, Sidford, and Vempala showed that $\widetilde O(n^2)$ membership queries suffice to implement an optimization oracle~\cite{lee2018efficient}.

Chakrabarti, Childs, Li, and Wu~\cite{chakrabarti2020quantum} and, independently, van Apeldoorn, Gily\'en, Gribling, and de Wolf~\cite{van2020convex} reduced the quantum query complexity to $\widetilde O(n)$. Their algorithms use quantum gradient estimation to implement a separation oracle with only polylogarithmically many membership queries. They then use a classical reduction that needs $\widetilde O(n)$ separation queries to perform optimization. Both works proved an $\Omega(\sqrt n)$ lower bound when an interior point of $K$ is known and asked whether a linear lower bound holds.

\begin{conjecture}[\cite{chakrabarti2020quantum,van2020convex}]
\label{conj:membership-lower-bound}
There is a family of $n$-dimensional convex bodies for which every quantum algorithm that implements an optimization oracle requires $\Omega(n)$ membership queries in the worst case. This remains true even when the algorithm is given an interior point of the body.
\end{conjecture}

This asks for the query cost of implementing $\mathsf{OPT}_K$ from $\mathsf{MEM}_K$. In the quantum setting, $\mathsf{SEP}_K$ can be implemented using $\widetilde O(1)$ membership queries~\cite{chakrabarti2020quantum,van2020convex}. We focus on lower bounds for implementing optimization directly from membership. Since separation can be implemented from membership with polylogarithmic overhead, our bounds also give lower bounds on the number of separation queries, weaker by at most a polylogarithmic factor.

Beyond optimization, there has been significant interest in quantum algorithms for linear algebraic tasks. Typical black-box query bounds for linear algebra problems operate in a sparse matrix query model. However, there has also been substantial interest in understanding query complexity in the matrix-vector product (MVP) model. Sun, Woodruff, Yang, and Zhang initiated a systematic classical study of matrix--vector product query complexity~\cite{sun2021querying}. Braverman, Hazan, Simchowitz, and Woodworth subsequently proved tight bounds for largest-eigenvector computation and linear regression at polynomial accuracy~\cite{braverman2020gradient}. Childs, Hung, and Li initiated the corresponding quantum study and proved linear query lower bounds over finite fields for problems including trace and determinant computation, rank testing, and solving linear systems~\cite{childs_et_al:LIPIcs.ICALP.2021.55}. They asked whether analogous quantum lower bounds could be proved for real matrices, where the finite-field polynomial arguments do not directly apply.

\subsection{Contributions}

Our main result proves Conjecture~\ref{conj:membership-lower-bound} above. The hard instances for our lower bound are linear optimization problems over centered ellipsoids.

\begin{theorem}[Exactly feasible optimization lower bound, informal]
\label{thm:intro-exactly-feasible}
For every $n\ge2$, there is an explicit family of centered $n$-dimensional ellipsoids $K$, each satisfying $B_2(0,r)\subseteq K\subseteq B_2(0,1)$ with $r=\Theta(n^{-1})$, such that any quantum algorithm that solves \eqref{eq:canonical-convex-optimization} for every ellipsoid in the family and every linear function $f(x)=-c^Tx$, where $\norm c=1$ and $\eps=1/(200\sqrt n)$, must make at least
\[
 \frac n{500}
\]
membership queries in the worst case.
\end{theorem}

The same bound holds for every smaller $\eps$, since a more accurate algorithm also satisfies the weaker guarantee. The first version of this paper proved a lower bound of order $n/(\log n\,\log\log n)$ for rescaled copies of these ellipsoids with $\eps=\Theta(n^{-2})$, using a reduction based on inverse iteration. We retain that argument in Section~\ref{sec:inverse-iteration-argument}, and Corollary~\ref{cor:direct-original-scales} shows that the linear bound also holds for those families.

We also prove a version of the lower bound, for a weaker task, that admits approximately feasible output.

\begin{theorem}[Approximately feasible optimization lower bound, informal]
\label{thm:intro-approximately-feasible}
The same lower bound holds for the same family if \eqref{eq:canonical-convex-optimization} is relaxed as follows: for $\delta=\Theta(n^{-3/2})$, the algorithm may return any $y$ satisfying
\[
 \dist(y,K)\le\delta,
 \qquad
 f(y)\le\min_{z\in K^{-\delta}}f(z)+\delta.
\]
Here $K^{-\delta}$ is the set of points $z$ for which $B_2(z,\delta)\subseteq K$.
\end{theorem}

The optimization lower bound is built on determinant hardness. We first prove the determinant lower bound in the continuous matrix phase-query model and then adapt the same Fourier-rank argument to membership queries for the hard ellipsoids. The same construction also gives an $\Omega(n)$ phase-query lower bound for estimating the minimum eigenvalue of a real symmetric matrix to additive accuracy $\Theta(n^{-2})$.

\begin{theorem}[Continuous matrix phase determinant and eigenvalue bounds, informal]
\label{thm:intro-mv}
In the continuous matrix phase-query model, computing the determinant of a real $n\times n$ matrix requires at least $n/2$ queries in the worst case. Estimating the minimum eigenvalue of a real symmetric $n\times n$ matrix to additive error $\Theta(n^{-2})$ requires $\Omega(n)$ queries.
\end{theorem}

The same hard matrices also yield a complementary result in a different oracle setting. In Section~\ref{sec:strongly-convex-quadratics}, we study unconstrained first-order optimization of smooth and strongly convex quadratic functions in the ideal continuous-register gradient phase-query model.
\begin{theorem}[First-order lower bound for strongly convex quadratics, informal]
\label{thm:intro-strongly-convex-gradient}
For every $n\ge1$ and $\kappa\ge1$, put $q:=\min\{\sqrt\kappa,n\}$. There is an explicit family of $n$-dimensional quadratic functions, each $1$-strongly convex and $\kappa$-smooth, such that any quantum algorithm that, with probability at least $2/3$, returns a point with additive objective error at most $1/200$ requires
\[
 \Omega(q)
\]
queries to the gradient phase oracle in the worst case. Thus the dependence on $\min\{\sqrt\kappa,n\}$ is optimal up to a logarithmic factor.
\end{theorem}

For the hard family above, Nesterov's accelerated gradient descent and conjugate gradient give $O(\sqrt\kappa\log(2+\kappa))$ and $O(n)$ upper bounds, respectively~\cite{nesterov_acceleration}. Thus our theorem rules out any quantum speedup in $q=\min\{\sqrt\kappa,n\}$ beyond a logarithmic factor. The first version of this paper proved this bound up to a factor $\log q\log\log q$; Section~\ref{sec:quadratic-inverse-iteration} retains that argument. Combined with the accuracy-dependent quantum lower bounds of Garg, Kothari, Netrapalli, and Sherif~\cite{garg2020no,garg2021near}, this establishes quantum optimality of the principal classical parameter dependences in the respective nonsmooth, smooth convex, and constant-accuracy strongly convex regimes. Because the bounds use different hard families, they do not imply a joint $\Omega(\sqrt\kappa\log(1/\eps))$ lower bound for strongly convex optimization.

The hard inputs for the determinant problem are Haar-random matrices $Q\in \mathsf{O}(n)$. The corresponding hard inputs for eigenvalue estimation are their symmetric parts
\[
 S_Q=\frac{Q+Q^T}{2}.
\]
Our lower bounds are based on a variant of the polynomial method that is based on Fourier-rank. This variant is particularly well-suited to matrix phase oracles and gives a convenient alternative to the well studied polynomial method~\cite{beals2001quantum,buhrman2002complexity,bun2018polynomial} from quantum query complexity. This method is encapsulated in Theorem~\ref{thm:Fourier-rank-method} and may be of independent interest for the establishment of further quantum lower bounds for linear algebraic problems.

\paragraph{Why previous approaches do not extend directly.}
The existing $\Omega(\sqrt n)$ lower bounds of Chakrabarti, Childs, Li, and Wu~\cite{chakrabarti2020quantum} and van Apeldoorn, Gily\'en, Gribling, and de Wolf~\cite{van2020convex} both use the hypercube as a hard instance, with reductions to search with wildcards and unstructured search, respectively. These reductions naturally capture the quadratic quantum speedup available for search, but for the same reason they do not provide a clear route to a near-linear lower bound. Such a lower bound would rule out any polynomial quantum speedup in the dimension, up to logarithmic factors. Results of this kind are relatively uncommon in quantum query complexity; familiar examples include binary search, sorting, and parity~\cite{hoyer2002quantum,beals2001quantum}. Within black-box convex optimization, the closest prior results are due to Garg, Kothari, Netrapalli, and Sherif, who rule out polynomial quantum speedups in the accuracy dependence of first-order optimization under several smoothness assumptions~\cite{garg2020no,garg2021near}. Their bounds, however, require the dimension to grow as a function of the target accuracy. For example, their nonsmooth lower bound is $\Omega(1/\eps^2)$ when $n=\Omega(1/\eps^4)$, which gives only an $\Omega(\sqrt n)$ lower bound when expressed in terms of the dimension. It therefore does not resolve the high-accuracy, dimension-dependent question considered here.

The intuition behind a near-linear lower bound is that the best classical algorithms in this regime proceed through a sequence of adaptively chosen cuts~\cite{vaidya1996new,lee2015faster,lee2018efficient}, and therefore appear inherently sequential. Sequentiality alone, however, does not preclude a quantum speedup: apparently sequential procedures, including some dynamic-programming algorithms, can still admit polynomial quantum improvements~\cite{ambainis2019quantum}. Classical matrix--vector lower bounds provide another possible starting point. In particular, Braverman, Hazan, Simchowitz, and Woodworth prove a linear lower bound for linear regression using a delicate argument that conditions on the classical transcript of previous queries and responses~\cite{braverman2020gradient}. A quantum algorithm may query coherently without producing such a transcript, and measuring the query history would generally change the algorithm. Their conditioning argument therefore does not transfer directly to the quantum setting. Our proof avoids these obstacles by reducing determinant prediction to convex optimization and using Fourier rank to control the dependence of a fully coherent query algorithm on the hidden matrix.

\paragraph{Concurrent work.}
In the later stages of this project, we became aware of independent and concurrent upcoming work by Andrew M. Childs~\cite{childs2026quantumlowerboundsconvex} that obtains similar asymptotic lower bounds for convex optimization and related real matrix--vector query problems. The two approaches were developed independently. Our technical analysis shares some commonalities with~\cite{childs2026quantumlowerboundsconvex} but the overall proofs have notable differences. As in this paper, that work also reports the use of Large Language Models in developing its proofs. On September 8, the first version of this paper~\cite{augustino2026nearoptimalquantumlowerbounds} appeared on arXiv, on the same day as the first public version of the work by Childs~\cite{childs2026quantumlowerboundsconvex}. We provide a more detailed comparison with~\cite{childs2026quantumlowerboundsconvex} in Section~\ref{sec:related-work}.

\paragraph{Proof overview.}
The proof has two main layers. First, we develop a Fourier-rank analogue of the polynomial method and use it to prove determinant lower bounds for both matrix phase queries and membership queries. Second, we reduce determinant prediction to convex optimization over a family of ellipsoids.

\paragraph{The Fourier-rank polynomial method.}
The main technical idea is a continuous analogue of the polynomial method for quantum query complexity~\cite{beals2001quantum}. In the usual polynomial method, one tracks the degree of the input variables appearing in an algorithm's amplitudes. Here the hidden input is a real matrix $Q$. We instead expand functions of $Q$ in Fourier exponentials
\[
 \chi_\Xi(Q):=e^{i\Tr(\Xi^TQ)}
\]
and measure the complexity of the frequency matrix $\Xi$ by its rank. We call this quantity its Fourier rank.

This measure is particularly compatible with the matrix phase-query model. One query contributes a factor
\[
 e^{iu^TQv}=\chi_{uv^T}(Q),
\]
whose frequency matrix $uv^T$ has rank one. Multiplying Fourier exponentials adds their frequency matrices and since matrix rank is subadditive,  increases the Fourier rank by at most one. It follows that after $T$ queries, every amplitude has Fourier rank at most $T$. An output probability pairs two amplitudes through the Born rule, so its Fourier rank is at most $2T$. These are the analogues of the degree-$T$ and degree-$2T$ bounds for amplitudes and probabilities in the ordinary polynomial method.

To obtain a lower bound, we show that the determinant is orthogonal to every frequency of rank less than $n$. If $\rank(\Xi)<n$, there is a reflection $J$ satisfying
\[
 J^T\Xi=\Xi,
 \qquad
 \det J=-1.
\]
Thus $\chi_\Xi(JQ)=\chi_\Xi(Q)$, while $\det(JQ)=-\det Q$. Haar invariance then forces the correlation between $\chi_\Xi$ and $\det Q$ to vanish. Therefore an output probability of Fourier rank less than $n$ cannot distinguish the two determinant components, which gives the $n/2$ query lower bound.

This conclusion is stronger than a standard bounded-error lower bound. With fewer than $n/2$ queries, an algorithm cannot obtain even an arbitrarily small Haar-average advantage over random guessing. In this sense, the determinant plays a role analogous to parity in the ordinary polynomial method: parity is invisible to low-degree polynomials, while the determinant is invisible to low-rank matrix frequencies. More generally, the method suggests notions of sign and approximate Fourier rank analogous to threshold degree and approximate degree. We leave the development of these measures and their application to other continuous query problems for future work.

\paragraph{Minimum-eigenvalue estimation.}
To obtain the eigenvalue lower bound, we replace $Q$ by its symmetric part
\[
 S_Q=\frac{Q+Q^T}{2}.
\]
On the negative determinant component, $S_Q$ has an eigenvalue $-1$. With high probability on the positive component, all eigenvalues are separated from $-1$ by $\Theta(n^{-2})$. An estimate of $\lambda_{\min}(S_Q)$ to this accuracy would therefore predict $\det Q$. Since one phase query to $S_Q$ can be implemented using two phase queries to $Q$, the determinant bound gives an $\Omega(n)$ lower bound.

\paragraph{Membership queries.}
The membership proof shows that the Fourier-rank method also applies to a discontinuous Boolean oracle. For each queried point, membership in the hard ellipsoid is equivalent to a threshold
\[
 q^TQq\le \tau.
\]
Here the vector $q$ and scalar $\tau$ are determined by the queried point. The corresponding Boolean phase is discontinuous, but Fej\'er approximation expresses it as an $L^2$ limit of trigonometric polynomials whose exponential terms have frequency matrices $tqq^T$. These matrices have rank at most one. A membership query therefore increases the Fourier rank of an amplitude by at most one, just as a matrix phase query does. Reflection orthogonality then gives the same $n/2$ determinant lower bound for membership queries to the hard ellipsoids.

\paragraph{From optimization to determinant prediction.}
It remains to show that an accurate optimizer can predict the determinant. For each hidden matrix $Q$, we construct a positive-definite matrix $A_Q$ and the ellipsoid
\[
 K_{A_Q}:=\{x:x^TA_Qx\le1\}.
\]
The maximum of $b^Tx$ over this ellipsoid equals $(b^TA_Q^{-1}b)^{1/2}$ for every unit vector $b$. An accurate optimizer therefore estimates this quantity.

When $\det Q=-1$, the matrix $A_Q$ has the eigenvalue $2g_n$, where $g_n=\Theta(n^{-2})$ is the spectral scale used in the construction. Let $b$ be a normalized vector of independent random signs. With constant probability, $b$ has squared overlap of order $1/n$ with the corresponding eigenvector, and then $b^TA_Q^{-1}b$ is at least of order $n$. When $\det Q=+1$, few eigenangles of $Q$ lie near $\pi$, and a bound on their expected density shows that the average of $b^TA_Q^{-1}b$ is at most a small constant multiple of $n$. A test that compares the optimal value with a fixed threshold therefore predicts $\det Q$ with a small but positive advantage over random guessing. The test uses $125$ optimizer invocations with the same objective. Since fewer than $n/2$ membership queries give no advantage at all, each optimizer invocation must use at least $n/500$ membership queries.

\paragraph{The argument from the first version.}
The first version of this paper used a different reduction, which we retain in Section~\ref{sec:inverse-iteration-argument}. The unique maximizer of $b^Tx$ over $K_{A_Q}$ is a normalized copy of $A_Q^{-1}b$, and its objective value supplies the missing normalization. A sufficiently accurate optimizer therefore gives an approximate application of $A_Q^{-1}$.

With high probability on the negative determinant component, $A_Q$ has one exceptional eigenvalue $2g_n$, while all its other eigenvalues are at least $3g_n$. Repeated approximate inverse applications amplify the corresponding eigendirection. After $O(\log n)$ inverse-iteration steps, the iterate is nearly aligned with this direction. With high probability on the positive component, every eigenvalue of $A_Q$ is at least $3g_n$. A final optimization call produces an objective value that lies on different sides of a fixed threshold in the two cases.

Each inverse step succeeds only with constant probability. We repeat and cluster the optimizer outputs, introducing an additional $O(\log\log n)$ factor. The complete determinant predictor uses $O(\log n\log\log n)$ optimizer invocations. Combining this reduction with the $n/2$ membership-query lower bound gives
\[
 \Omega\!\left(\frac{n}{\log n\,\log\log n}\right)
\]
membership queries per optimizer invocation.

\subsection{Related work}
\label{sec:related-work}

The oracle framework for convex optimization goes back to Nemirovskii and Yudin and to Gr\"otschel, Lov\'asz, and Schrijver~\cite{nemirovskii1983problem,grotschel2012geometric}. Vaidya gave the first cutting-plane method that uses nearly linear-in-$n$ separation queries~\cite{vaidya1996new}. Lee, Sidford, and Wong obtained the same query bound with fewer arithmetic operations~\cite{lee2015faster}. Lee, Sidford, and Vempala later showed that $\widetilde O(n^2)$ membership queries suffice~\cite{lee2018efficient}.

The quantum algorithms of Chakrabarti, Childs, Li, and Wu and of van Apeldoorn, Gily\'en, Gribling, and de Wolf use quantum gradient estimation to reduce the membership-query complexity to $\widetilde O(n)$~\cite{chakrabarti2020quantum,van2020convex}. Our lower bound matches this dependence on $n$ up to logarithmic factors in the model studied here.

Garg, Kothari, Netrapalli, and Sherif proved quantum lower bounds for first-order convex optimization over a range of smoothness assumptions~\cite{garg2020no,garg2021near}. These results describe the regime in which the accuracy controls the complexity and the dimension is large. Parallel convex optimization has a different set of tradeoffs~\cite{duchi2012randomized,bubeck2019complexity}. Our result concerns sequential membership queries in the high-accuracy regime.

Matrix--vector query models have been studied for determinant computation, eigenvalue estimation, principal component analysis, regression, and trace estimation~\cite{dorn2009quantum,childs_et_al:LIPIcs.ICALP.2021.55,simchowitz2018tight,braverman2020gradient,sun2021querying,wimmer2014optimal,meyer2021hutch++}. Our matrix results use the continuous phase formulation. Inverse iteration and shift-and-invert are standard tools in this area~\cite{saad2003iterative,saad2011numerical,garber2016faster}. Our main lower bound uses a single random objective. The argument from the first version uses a version of inverse iteration that tolerates the error produced by the optimization oracle.

\paragraph{Comparison with concurrent work.}
Childs~\cite{childs2026quantumlowerboundsconvex} independently obtained nearly linear quantum lower bounds for convex optimization and for real matrix--vector query problems. The lower-bound techniques in the two papers are substantially the same. Both show that, after $t$ phase queries, every output probability has Fourier frequencies of rank at most $2t$. Both then use the fact that the determinant vanishes on matrices of rank below $n$, so a suitable target is uncorrelated with every such frequency. Childs expresses this through dual witnesses obtained by applying the Cayley operator $\det(\partial_X)$ to a Gaussian density. We express it through a reflection that fixes each frequency matrix of rank below $n$ and reverses the sign of $\det Q$ for Haar-random orthogonal $Q$.

The two papers differ in what they prove with these tools. Childs's Gaussian witnesses apply to a broader range of matrix functions. They give $\Omega(n/\log n)$ bounds for the trace, for the magnitude of the determinant, and for the first column of the inverse of a general or positive definite real matrix, none of which we treat. His results are stated in a discretized phase-query model with explicit error bounds for every simulation. His optimization instances are quadratics over the unit ball, accessed through function evaluations, and he proves an $\Omega(n/\log n)$ lower bound at accuracy of order $n^{-49}$. The first version of this paper~\cite{augustino2026nearoptimalquantumlowerbounds}, which appeared on arXiv on the same day, worked instead in the membership-oracle setting of Conjecture~\ref{conj:membership-lower-bound}. It proved an $\Omega(n/(\log n\log\log n))$ bound by inverse iteration (Section~\ref{sec:inverse-iteration-argument}), at accuracy of order $n^{-4}$ relative to the outer radius. Childs's bound has one fewer logarithmic factor, while the first version holds at a much larger accuracy. The first version also proved lower bounds that Childs does not consider: an analogous bound for approximately feasible optimization, an $\Omega(n)$ bound for estimating the minimum eigenvalue to accuracy $\Theta(n^{-2})$, and a bound for strongly convex quadratics that is linear in $\min\{\sqrt\kappa,n\}$ up to logarithmic factors. The present version uses a single random objective in place of inverse iteration. It removes all logarithmic factors from the optimization bounds, proving $\Omega(n)$ membership queries at accuracy $1/(200\sqrt n)$ relative to the outer radius (Corollary~\ref{cor:direct-normalized}), and the same bound for approximately feasible optimization. It also proves an $\Omega(\min\{\sqrt\kappa,n\})$ bound for strongly convex quadratics at accuracy $1/200$, which matches accelerated gradient descent and conjugate gradient up to a logarithmic factor (Theorem~\ref{thm:direct-gradient-lower-bound}).

The two formulations also have different conceptual strengths. Childs's method gives a general recipe for constructing witnesses. Any polynomial that vanishes on matrices of rank at most $r$, applied as a differential operator to a smooth density, produces a measure whose Fourier transform vanishes on those matrices. The density can then be chosen to correlate with the target function. This parallels the role of dual polynomials in the classical polynomial method, and it makes the method applicable to targets without any special symmetry. Our formulation, which is the same in both versions of this paper, has two features of a different kind. First, it separates the cost of a single query from the choice of target. Theorem~\ref{thm:Fourier-rank-method} requires only that each query increase the Fourier rank by a bounded amount. This allowed us to treat Boolean membership queries directly: the membership phase is a discontinuous function of $Q$, but it is an $L^2$ limit of exponentials whose frequencies have rank one. Second, the reflection argument gives exact orthogonality, so an algorithm with fewer than $n/2$ queries has no advantage at all in predicting $\det Q$. The first version did not need this property, because inverse iteration predicts the determinant with advantage $1/50$. The present version uses it to make any positive advantage sufficient. This matters because a single random objective predicts $\det Q$ with an advantage of only about $10^{-3}$, and exact orthogonality removes any need to compare this advantage with a threshold from the lower bound.

\subsection{Organization}

Section~\ref{sec:prelim} collects the notation used throughout the paper and defines the query models. Section~\ref{sec:matrix-lower-bounds} develops the Fourier-rank polynomial method and applies it to determinant and minimum-eigenvalue estimation. Section~\ref{sec:determinant-to-optimization} constructs the hard ellipsoids, establishes determinant hardness for exact membership queries, and shows how the optimal value of one random linear objective predicts the determinant. Section~\ref{sec:optimization-lower-bounds} proves the linear lower bounds for exactly and approximately feasible optimization. Section~\ref{sec:strongly-convex-quadratics} applies the same construction to first-order optimization of smooth and strongly convex quadratics. Section~\ref{sec:inverse-iteration-argument} presents the inverse-iteration argument from the first version of this paper, which proves both results up to a factor of order $\log n\log\log n$. Appendix~\ref{app:continuous-membership} extends the membership obstruction to continuous query registers.

\section{Preliminaries}
\label{sec:prelim}

We use two query models. The determinant and eigenvalue lower bounds use a continuous matrix phase oracle. The optimization lower bound uses a binary membership oracle on a finite set of represented points. We define both models below and comment on their relationship.

\subsection{Basic notation}

We first collect the geometric and matrix notation used in the proof.
Unless a base is displayed explicitly, $\log$ denotes the natural logarithm. This does not affect any of the asymptotics but does impact the explicit constants in some of the proofs.

For $x\in\R^n$ and $r>0$, let
\[
 B_2(x,r):=\{y\in\R^n:\norm{y-x}_2\le r\}.
\]
A convex body is a compact convex set with nonempty interior. As is standard in the oracle model for convex optimization, we assume throughout the optimization results that the algorithm is given parameters $0<r\le R$ such that
\begin{equation}
 B_2(0,r)\subseteq K\subseteq B_2(0,R).
 \label{eq:standard-ball-assumption}
\end{equation}
See~\cite{nemirovskii1983problem,grotschel2012geometric} for this standard setup.

For a nonempty set $K\subseteq\R^n$, define
\[
 \dist(x,K):=\inf_{y\in K}\norm{x-y}_2.
\]
For $\delta>0$, define
\begin{align}
 K^{+\delta}&:=\{x:\dist(x,K)\le\delta\},
 \\ 
 K^{-\delta}&:=\{x:B_2(x,\delta)\subseteq K\}.
 \label{eq:offsets}
\end{align}
For compact $K$, write
\[
 h_K(c):=\max_{x\in K}c^Tx.
\]
$K^{+\delta}$ allows points up to distance $\delta$ outside $K$. The set $K^{-\delta}$ removes the points within distance $\delta$ of the boundary. These sets are used in the standard definitions of approximate geometric oracles~\cite{grotschel2012geometric,van2020convex}.

We write $\mathbb S^{n-1}:=\{x\in\R^n:\norm{x}_2=1\}$ for the unit sphere and $\mathsf{O}(n):=\{Q\in\R^{n\times n}:Q^TQ=I\}$ for the orthogonal group.

For a matrix $A$, $\Tr(A)$ denotes its trace, $\rank(A)$ its rank, and $\col(A)$ its column space. A real symmetric matrix $A$ is positive definite if $x^TAx>0$ for every nonzero $x$. We order its eigenvalues as
\[
 \lambda_1(A)\ge\cdots\ge\lambda_n(A),
 \qquad
 \lambda_{\max}(A):=\lambda_1(A),
 \qquad
 \lambda_{\min}(A):=\lambda_n(A).
\]
For symmetric matrices $A$ and $B$, we write $A\succeq B$ when $x^T(A-B)x\ge0$ for every $x$. We also use
\[
 \norm{A}_{\rm op}:=\sup_{\norm{x}_2=1}\norm{Ax}_2,
 \qquad
 \norm{x}_A:=\sqrt{x^TAx}
\]
for the operator norm and, when $A$ is positive definite, the $A$-norm.

\subsection{Defining a discretized membership oracle}

For membership queries, we use the finite-register model of van Apeldoorn, Gily\'en, Gribling, and de Wolf~\cite{van2020convex}. Each coordinate is represented with finitely many bits. This leads to the following model.

\begin{definition}[Previously fixed finite grid]
\label{def:grid}
Before the body $K$ is chosen, fix a finite set $\mathcal G\subseteq\R^n$. Its elements label the basis states of the membership-query register.
\end{definition}

The grid can be made arbitrarily fine. For example, one may intersect $h\mathbb Z^n$ with a known box containing $K$ and choose $h>0$ as small as desired. Our lower bound holds for every finite grid fixed in advance and does not depend on its spacing. Finiteness allows the determinant proof to analyze the membership phase separately at each possible query point \footnote{This analytical need for finiteness can be avoided, at the cost of more involved arguments. This is discussed further in Section~\ref{sec:matrix-lower-bounds} with the corresponding proofs given in Appendix~\ref{app:continuous-membership}.}.

\begin{definition}[Exact coherent membership]
\label{def:exact-membership}
For $K\subseteq\R^n$ and a finite grid $\mathcal G$ fixed in advance, the exact membership oracle acts on a query point $x\in\mathcal G$ and an answer bit $b$ as
\begin{equation}
 U_K\ket{x}\ket{b}
 =\ket{x}\ket{b\oplus\1[x\in K]},
 \qquad x\in\mathcal G,
 \label{eq:exact-membership}
\end{equation}
where boundary points are counted as members of $K$.
\end{definition}

\begin{definition}[Weak coherent membership]
\label{def:weak-membership}
Fix $\delta_{\rm mem}>0$. A weak membership oracle is determined by a function $g:\mathcal G\to\{0,1\}$ such that
\[
 g(x)=1\Longrightarrow x\in K^{+\delta_{\rm mem}},
 \qquad
 g(x)=0\Longrightarrow x\notin K^{-\delta_{\rm mem}}.
\]
The oracle acts as $U_g\ket{x}\ket{b}=\ket{x}\ket{b\oplus g(x)}$. Thus either answer is allowed when $x$ lies in the promise region between $K^{-\delta_{\rm mem}}$ and $K^{+\delta_{\rm mem}}$.
\end{definition}

Unless stated otherwise, all membership queries below are exact.

\begin{remark}[Weak membership]
The lower bounds proved below also apply, with no query overhead, to algorithms required to work for every weak membership oracle allowed by Definition~\ref{def:weak-membership}. Indeed, exact membership is one admissible choice for every $\delta_{\rm mem}>0$: the function $g_K(x)=\1[x\in K]$ is valid because
\[
 K^{-\delta_{\rm mem}}\subseteq K\subseteq K^{+\delta_{\rm mem}}.
\]
This observation uses the requirement that the algorithm work for every admissible weak oracle. If a model fixes one particular convention for answers near the boundary, a separate reduction may be needed.
\end{remark}

\subsection{Optimization guarantees}
\label{sec:optimization-guarantees}

Both definitions below are stated for maximizing the linear objective $c^Tx$, which is the problem in \eqref{eq:canonical-convex-optimization} with $f(x)=-c^Tx$.

\begin{definition}[Exactly feasible optimization]
\label{def:exactly-feasible-opt}
An exactly feasible optimizer of accuracy $\eps>0$ receives a unit objective $c\in\mathbb S^{n-1}$ and returns a classical vector $\widehat x$.  With probability at least $2/3$ it must satisfy
\begin{equation}
 \widehat x\in K,
 \qquad
 c^T\widehat x\ge h_K(c)-\eps.
 \label{eq:exactly-feasible-output}
\end{equation}
\end{definition}

\begin{definition}[Approximately feasible optimization]
\label{def:approximately-feasible-opt}
Fix a tolerance $0<\delta<r$, where $r$ is the known inner radius in \eqref{eq:standard-ball-assumption}. Then $K^{-\delta}$ is nonempty. An approximately feasible optimizer with tolerance $\delta$ receives a unit objective $c$ and returns a classical vector $y$. With probability at least $2/3$ it must satisfy
\begin{equation}
 \dist(y,K)\le\delta,
 \qquad
 c^Ty\ge h_{K^{-\delta}}(c)-\delta.
 \label{eq:approximately-feasible-output}
\end{equation}
\end{definition}

In the exactly feasible problem, the output must lie in $K$ and is compared with the optimum over $K$. In the approximately feasible problem, the output may move outward by $\delta$, while the comparison set moves inward by $\delta$. This is the standard approximately feasible optimization guarantee~\cite{grotschel2012geometric,van2020convex}.
An approximately feasible output may lie outside $K$. The following lemma compares its objective value with the optimum over $K$ and supplies a nearby feasible point.

\begin{lemma}[Approximately feasible output geometry]
\label{lem:approximately-feasible-geometry}
Suppose $B_2(0,r)\subseteq K\subseteq B_2(0,R)$ and $0<\delta<r$.  If $y$ satisfies \eqref{eq:approximately-feasible-output} for a unit objective $c$, then there is $x\in K$ with $\norm{x-y}\le\delta$ and
\begin{equation}
 c^Tx\ge h_K(c)-D,
 \qquad
 D:=\delta\left(\frac Rr+2\right).
 \label{eq:approximately-feasible-effective-error}
\end{equation}
Moreover,
\begin{equation}
 h_K(c)-D\le c^Ty\le h_K(c)+\delta.
 \label{eq:approximately-feasible-objective-interval}
\end{equation}
\end{lemma}

\begin{proof}
Convexity and $B_2(0,r)\subseteq K$ imply
\[
 \left(1-\frac\delta r\right)K+\delta B_2(0,1)\subseteq K,
\]
so $(1-\delta/r)K\subseteq K^{-\delta}$. The objective guarantee in \eqref{eq:approximately-feasible-output} therefore gives
\[
 c^Ty\ge\left(1-\frac\delta r\right)h_K(c)-\delta.
\]

\subsection{The continuous matrix phase-query model}

The determinant and eigenvalue results use the continuous phase formulation of matrix--vector access studied by Childs, Hung, and Li~\cite{childs_et_al:LIPIcs.ICALP.2021.55}. On the Hilbert space $L^2(\R^n\times\R^n;\C)$, with query labels $(v,u)$, define
\begin{equation}
 (\mathcal P_Q\psi)(v,u)
 =e^{iu^TQv}\psi(v,u).
 \label{eq:continuous-phase}
\end{equation}
We call one application of $\mathcal P_Q$ or $\mathcal P_Q^{-1}$ a matrix phase query. The phase has the form
\[
 u^TQv=\Tr((uv^T)^TQ),
\]
and the frequency matrix $uv^T$ has rank one.

\begin{remark}[Ideal continuous translation oracle]\label{rem:ideal-continuous-translation}
In the same ideal continuous-register model, one may instead define the translation oracle
\begin{equation}
 (\mathcal O_Q\psi)(v,y)=\psi(v,y-Qv).
 \label{eq:continuous-mv}
\end{equation}
Use the Fourier-transform convention
\[
 (\mathcal F_yf)(u):=(2\pi)^{-n/2}
 \int_{\R^n}e^{iu^Ty}f(y)\,dy.
\]
Then a change of variables gives the exact identity
\begin{equation}
 \mathcal P_Q=(I\otimes\mathcal F_y)\mathcal O_Q
 (I\otimes\mathcal F_y^{-1}).
 \label{eq:mv-phase-equivalence}
\end{equation}
This identity holds in the ideal $L^2$ model. A finite-precision standard-basis encoding of $Qv$ defines a different oracle and need not produce the bilinear phase in \eqref{eq:continuous-phase}. The matrix lower bounds below are proved directly for $\mathcal P_Q$.
\end{remark}

\begin{remark}[Separation of query models]
The optimization lower bound uses only the finite binary membership oracle in Definition~\ref{def:exact-membership}. It neither invokes nor simulates the matrix phase oracle $\mathcal P_Q$. Conversely, the determinant and eigenvalue results use only $\mathcal P_Q$. The two arguments share a Fourier-rank obstruction, not an oracle simulation.
\end{remark}

\subsection{Fourier approximations}
\label{sec:Fejer-prelim}

We use one standard fact about Fej\'er means. Let $h$ be a bounded, integrable, two-periodic function on $[-1,1]$, with Fourier coefficients
\[
 \widehat h(k):=\frac12\int_{-1}^1h(t)e^{-i\pi kt}\,dt.
\]
Its $N$th Fej\'er mean is the weighted Fourier sum
\begin{equation}
 (F_Nh)(t):=\sum_{k=-N}^N
 \left(1-\frac{|k|}{N+1}\right)\widehat h(k)e^{i\pi kt}.
 \label{eq:Fejer-mean-prelim}
\end{equation}
Unlike an ordinary partial Fourier sum, this is an average of the first $N+1$ partial sums. The Fej\'er kernel is nonnegative and has unit mass. Consequently,
\[
 \norm{F_Nh}_\infty\le\norm{h}_\infty,
\]
and $F_Nh(t)\to h(t)$ at every continuity point of $h$~\cite{katznelson2004introduction}. In particular, if $h$ has only finitely many jumps per period, then the convergence holds almost everywhere and, by dominated convergence, in $L^2$. The same statement applies on any bounded interval after rescaling it to $[-1,1]$.

\section{Lower bounds on determinant and eigenvalue estimation}
\label{sec:matrix-lower-bounds}

This section first develops the Fourier-rank polynomial method as a general query lower-bound framework. We then apply it to the determinant and reduce determinant prediction to minimum-eigenvalue estimation.

\subsection{The Fourier-rank polynomial method}
\label{sec:mv-lower-bounds}

We begin with the standard quantum query algorithm model. After purifying classical randomness and intermediate measurements, a $T$-query algorithm has the form
\begin{equation}
 \ket{\psi_T(Q)}
 =U_T\mathcal O_{Q,T}U_{T-1}\cdots
 U_1\mathcal O_{Q,1}U_0\ket{\psi_0}.
 \label{eq:standard-query-algorithm}
\end{equation}
The initial state $\ket{\psi_0}$, the inter-query unitaries $U_0,\ldots,U_T$, and the final measurement are independent of the hidden matrix $Q$. Each $\mathcal O_{Q,j}$ is an allowed query or inverse query. This representation also covers adaptive algorithms: an intermediate measurement can be stored coherently in an ancilla, and later operations can be controlled by that ancilla before all measurements are deferred to the end. A query slot may likewise be controlled by the stored history; its Fourier-rank cost is at most the largest cost among its branches. Algorithms using fewer than $T$ queries on some branches can be padded with identity queries.

The method tracks the dependence of \eqref{eq:standard-query-algorithm} on $Q$. Instead of polynomial degree, we use the matrix rank of a Fourier frequency.

Throughout this subsection, $Q$ ranges over $\mathsf{O}(n)$ and $\mu_n$ denotes normalized Haar measure on $\mathsf{O}(n)$.

\begin{definition}[Fourier rank]
\label{def:fourier-rank}
For $\Xi\in\R^{n\times n}$, define
\begin{equation}
 \chi_\Xi(Q):=e^{i\Tr(\Xi^TQ)}.
 \label{eq:fourier-exponential}
\end{equation}
We call $\Xi$ the frequency matrix of $\chi_\Xi$ and $\rank\Xi$ its Fourier rank. For an integer $r\ge0$, define
\begin{equation}
 \mathcal V_r^{(n)}
 :=\overline{\operatorname{span}}^{\,L^2}
 \{\chi_\Xi:\rank\Xi\le r\}
 \subseteq L^2(\mathsf{O}(n),\mu_n).
 \label{eq:rank-space}
\end{equation}
For scalar outcome probabilities, define the corresponding $L^1$ closure
\begin{equation}
 \mathcal A_r^{(n)}
 :=\overline{\operatorname{span}}^{\,L^1}
 \{\chi_\Xi:\rank\Xi\le r\}
 \subseteq L^1(\mathsf{O}(n),\mu_n).
 \label{eq:probability-rank-space}
\end{equation}
Thus $\mathcal V_r^{(n)}$ is the rank-$r$ space used for amplitudes, whereas $\mathcal A_r^{(n)}$ is used for scalar outcome probabilities.
\end{definition}

This argument is a continuous-parameter analogue of the polynomial method for quantum queries~\cite{beals2001quantum}.  It tracks Fourier rank in place of polynomial degree.  The closure and approximation steps below are standard harmonic analysis~\cite{katznelson2004introduction}.

The following product and rank identities explain how these spaces grow:
\begin{equation}
 \chi_\Xi\chi_\Gamma=\chi_{\Xi+\Gamma},
 \qquad
 \chi_\Xi\overline{\chi_\Gamma}=\chi_{\Xi-\Gamma},
 \qquad
 \rank(\Xi\pm\Gamma)\le\rank\Xi+\rank\Gamma.
 \label{eq:frequency-algebra}
\end{equation}

The notation $\widehat\otimes$ denotes the completed Hilbert-space tensor product. Let $(Y,\nu)$ be a sigma-finite measure space and let $\mathcal K$ be a separable Hilbert space. Define
\begin{equation}
 \mathcal W_r(Y,\mathcal K)
 :=\mathcal V_r^{(n)}\widehat\otimes L^2(Y,\nu;\mathcal K).
 \label{eq:controlled-rank-space}
\end{equation}
Under the standard product-space identification, $\Psi\in\mathcal W_r(Y,\mathcal K)$ exactly when $\Psi(\cdot,y)\in\mathcal V_r^{(n)}\widehat\otimes\mathcal K$ for almost every $y$.

The notation serves three roles. The space $\mathcal V_r^{(n)}$ contains scalar amplitude functions of $Q$. The space $\mathcal A_r^{(n)}$ contains scalar outcome probabilities, for which $L^1$ convergence is sufficient. The space $\mathcal W_r(Y,\mathcal K)$ adds arbitrary query labels and workspace registers to the amplitude space. In each case, the subscript $r$ is the maximum matrix rank allowed in the Fourier frequencies. The proof uses only two bookkeeping rules: one query increases the amplitude rank by at most one, and forming a probability doubles the rank bound.

\paragraph{Fourier-rank cost of a query.}
A family of query oracles $\{\mathcal O_Q:Q\in \mathsf{O}(n)\}$ acting on a Hilbert space $\mathcal H$ has Fourier-rank cost at most $s$ if, for every $r\ge0$, both the query and its inverse map
\[
 \mathcal V_r^{(n)}\widehat\otimes\mathcal H
 \quad\text{into}\quad
 \mathcal V_{r+s}^{(n)}\widehat\otimes\mathcal H.
\]
The next lemma gives the basic example. A query controlled by a label $y$ has cost one whenever the corresponding frequency matrix has rank at most one for almost every $y$.

\begin{lemma}[Controlled rank-one phase growth]
\label{lem:controlled-phase-growth}
Let $\Gamma:Y\to\R^{n\times n}$ be measurable and satisfy $\rank\Gamma(y)\le1$ for almost every $y$. Define
\begin{equation}
 (P_\Gamma\Psi)(Q,y):=\chi_{\Gamma(y)}(Q)\Psi(Q,y).
 \label{eq:controlled-phase}
\end{equation}
Then $P_\Gamma$ is unitary and
\[
 P_\Gamma\mathcal W_r(Y,\mathcal K)
 \subseteq\mathcal W_{r+1}(Y,\mathcal K)
\]
for every $r\ge0$. The same statement holds for $P_\Gamma^{-1}$.
\end{lemma}

\begin{proof}
The multiplier in \eqref{eq:controlled-phase} is measurable and has modulus one, so $P_\Gamma$ is unitary. Fix a label $y$ for which $\rank\Gamma(y)\le1$. If
\[
 f(Q)=\sum_{j=1}^N c_j\chi_{\Xi_j}(Q),
 \qquad
 \rank\Xi_j\le r,
\]
then
\[
 \chi_{\Gamma(y)}(Q)f(Q)
 =\sum_{j=1}^N c_j\chi_{\Xi_j+\Gamma(y)}(Q),
 \qquad
 \rank(\Xi_j+\Gamma(y))\le r+1.
\]
Multiplication by $\chi_{\Gamma(y)}$ is an $L^2$ isometry. Approximation by finite Fourier sums and closedness of $\mathcal V_{r+1}^{(n)}$ therefore show that the $y$-fiber is mapped from $\mathcal V_r^{(n)}\widehat\otimes\mathcal K$ into $\mathcal V_{r+1}^{(n)}\widehat\otimes\mathcal K$. The multiplier also preserves the norm of every fiber. Integrating those norms over $Y$ gives the claimed inclusion in $\mathcal W_{r+1}(Y,\mathcal K)$. Replacing $\Gamma$ by $-\Gamma$ proves the statement for the inverse.
\end{proof}

Let $\mathcal H_{\rm alg}$ denote the Hilbert space of all query, workspace, and ancilla registers. The following theorem is the reusable core of the method. It separates the rank bookkeeping from both the formula for a particular oracle and the choice of target function.

\begin{theorem}[Fourier-rank polynomial method]
\label{thm:Fourier-rank-method}
\label{lem:mv-rank-growth}
Consider a $T$-query algorithm of the form \eqref{eq:standard-query-algorithm}. Suppose its $j$th query has Fourier-rank cost at most $s_j$, and set
\[
 R:=\sum_{j=1}^T s_j.
\]
The algorithm has a purification whose final state lies in
\[
 \mathcal V_R^{(n)}\widehat\otimes\mathcal H_{\rm alg},
\]
and every final outcome probability belongs to $\mathcal A_{2R}^{(n)}$.

Moreover, let $F:\mathsf{O}(n)\to\{-1,+1\}$ be a measurable target function such that
\begin{equation}
 \int_{\mathsf{O}(n)}F(Q)f(Q)\,d\mu_n(Q)=0
 \qquad
 \text{for every }f\in\mathcal A_{2R}^{(n)}.
 \label{eq:generic-target-orthogonality}
\end{equation}
If the algorithm outputs a prediction in $\{-1,+1\}$, its Haar-average success probability for predicting $F(Q)$ is exactly $1/2$.
\end{theorem}

\begin{proof}
The initial state has only the zero frequency. A $Q$-independent operation acts as the identity on the $L^2(\mathsf{O}(n),\mu_n)$ factor and therefore preserves every rank space. The definition of Fourier-rank cost and induction show that the final state has amplitude rank at most $R$. Intermediate measurements and classical randomness may be included in the purification.

Let $\Psi$ be the final state. There are finite expansions
\[
 \Psi^{(a)}(Q)=\sum_j\chi_{\Xi_{a,j}}(Q)\xi_{a,j},
 \qquad
 \rank\Xi_{a,j}\le R,
\]
that converge to $\Psi$ in $L^2(\mathsf{O}(n),\mu_n;\mathcal H_{\rm alg})$. For a final POVM effect $0\preceq E\preceq I$, their outcome probabilities are
\[
 p_E^{(a)}(Q)
 =\sum_{j,k}\chi_{\Xi_{a,k}-\Xi_{a,j}}(Q)
 \langle\xi_{a,j},E\xi_{a,k}\rangle.
\]
Every displayed frequency has rank at most $2R$. Since $\norm{E}_{\rm op}\le1$, Cauchy--Schwarz gives
\[
 \norm{p_E^{(a)}-p_E}_{L^1}
 \le\norm{\Psi^{(a)}-\Psi}_{L^2}
 \left(\norm{\Psi^{(a)}}_{L^2}+ \Big\|\Psi \Big\|_{L^2}\right)
 \longrightarrow0.
\]
Thus $p_E\in\mathcal A_{2R}^{(n)}$.

For the last claim, let $p_+(Q)$ be the probability that the algorithm outputs $+1$. Its Haar-average success probability is
\[
 \frac12+\frac12\int_{\mathsf{O}(n)}
 F(Q)\bigl(2p_+(Q)-1\bigr)\,d\mu_n(Q).
\]
Since $p_+\in\mathcal A_{2R}^{(n)}$ and the constant function belongs to $\mathcal A_0^{(n)}\subseteq\mathcal A_{2R}^{(n)}$, equation~\eqref{eq:generic-target-orthogonality} makes the integral zero.
\end{proof}

Theorem~\ref{thm:Fourier-rank-method} reduces a query lower bound to two tasks: bound the Fourier-rank cost of one query, and prove that the target is orthogonal to all frequencies below a given rank. This division is analogous to the ordinary polynomial method, where one first bounds the degree accumulated by a query algorithm and then proves that the target cannot be represented or approximated at that degree. The framework is not tied to one oracle: below, matrix phase queries have cost one, and Section~\ref{sec:membership-rank} shows that exact membership queries for the hard ellipsoids also have cost one.

The analogy also suggests notions of thresholded Fourier rank and approximate Fourier rank, corresponding to threshold degree and approximate degree. The determinant result below gives an especially strong obstruction: every function of Fourier rank below $n$ has zero Haar correlation with the determinant. Consequently, fewer than $n/2$ queries cannot obtain even an arbitrarily small average advantage over random guessing. We leave a systematic development of these notions and their application to other continuous query problems for future work.

\subsection{The determinant lower bound}

We now specialize the method to $F(Q)=\det Q$. The two determinant components are
\[
 \mathsf{O}^+(n):=\mathsf{SO}(n)=\{Q\in \mathsf{O}(n):\det Q=+1\},
 \qquad
 \mathsf{O}^-(n):=\{Q\in \mathsf{O}(n):\det Q=-1\}.
\]
Each component has $\mu_n$-measure $1/2$. We write $\mu_{n,+}$ and $\mu_{n,-}$ for the corresponding conditional probability measures.

The oracle side of the method is immediate. For the matrix phase oracle, take $Y=\R^n\times\R^n$, with labels $(v,u)$, and set $\Gamma(v,u)=uv^T$. Since $\rank(uv^T)\le1$, Lemma~\ref{lem:controlled-phase-growth} shows that each matrix phase query has Fourier-rank cost at most one.

It remains to establish the determinant-specific orthogonality property. A low-rank frequency cannot correlate with the determinant because a suitable reflection preserves the frequency and reverses the determinant.

\begin{lemma}[Reflection orthogonality]
\label{lem:reflection}
If $\rank\Xi<n$, then
\begin{equation}
 \int_{\mathsf{O}(n)}\det(Q)\chi_\Xi(Q)\,d\mu_n(Q)=0.
 \label{eq:reflection-zero}
\end{equation}
Consequently, $Q\mapsto\det Q$ is orthogonal to $\mathcal V_{n-1}^{(n)}$, and
\begin{equation}
 \int_{\mathsf{O}(n)}\det(Q)f(Q)\,d\mu_n(Q)=0
 \qquad
 \text{for every }f\in\mathcal A_{n-1}^{(n)}.
 \label{eq:reflection-L1}
\end{equation}
\end{lemma}

\begin{proof}
Choose a unit vector $w\in\col(\Xi)^\perp$ and let $J=I_n-2ww^T$. Then $\det J=-1$ and $J^T\Xi=\Xi$. Moreover,
\[
 \chi_\Xi(JQ)
 =e^{i\Tr(\Xi^TJQ)}
 =e^{i\Tr((J^T\Xi)^TQ)}
 =\chi_\Xi(Q).
\]
The Haar-measure-preserving change of variables $Q\mapsto JQ$ therefore changes the sign of the integrand in \eqref{eq:reflection-zero} while preserving its Fourier exponential. The integral equals its negative and is zero. Continuity of the $L^2$ inner product gives the conclusion for $\mathcal V_{n-1}^{(n)}$.

For the $L^1$ statement, define
\[
 \Lambda(f):=\int_{\mathsf{O}(n)}\det(Q)f(Q)\,d\mu_n(Q).
\]
Since $\abs{\det Q}=1$, we have $\abs{\Lambda(f)}\le\norm{f}_{L^1}$. Thus $\Lambda$ is continuous on $L^1(\mathsf{O}(n),\mu_n)$. It vanishes on every finite Fourier sum with frequency rank below $n$, so it also vanishes on their $L^1$ closure $\mathcal A_{n-1}^{(n)}$.
\end{proof}

\begin{theorem}[Continuous matrix phase determinant lower bound]
\label{thm:mv-determinant}
Any algorithm which predicts $\det Q$ for $Q\sim\mu_n$ with average success probability strictly greater than $1/2$, using the matrix phase oracle \eqref{eq:continuous-phase}, makes at least $n/2$ queries.
\end{theorem}

\begin{proof}
Suppose that $T<n/2$. Every query has Fourier-rank cost at most one, so $R=T$. Lemma~\ref{lem:reflection} makes the determinant orthogonal to $\mathcal A_{2T}^{(n)}\subseteq\mathcal A_{n-1}^{(n)}$. Theorem~\ref{thm:Fourier-rank-method} therefore makes the Haar-average success probability exactly $1/2$.
\end{proof}

The lower bound is distributional: every algorithm using fewer than $n/2$ queries has Haar-average success probability $1/2$. In particular, every such algorithm fails on some $Q\in \mathsf{O}(n)$. Therefore, the lower bound naturally holds for worst-case instances. Since orthogonal matrices are real matrices, this also gives a worst-case lower bound for determinant computation on general real matrices.

\begin{remark}[Finite query registers]
\label{rem:finite-phase-registers}
Lemma~\ref{lem:controlled-phase-growth} and Theorem~\ref{thm:Fourier-rank-method} hold for every sigma-finite label space $(Y,\nu)$, including any finite set of labels in $\R^n\times\R^n$ with counting measure, so Theorems~\ref{thm:mv-determinant} and~\ref{thm:min-eigenvalue} and Theorem~\ref{thm:direct-gradient-lower-bound} hold without change when the phase oracles $\mathcal P_Q$ and $\mathcal G_f$ act only on a finite set of query labels fixed independently of $Q$, such as the discretized phase oracles of~\cite{childs2026quantumlowerboundsconvex}; in particular, a gradient phase query at a label $(x,u)$ in such a set is the $Q$-independent phase in \eqref{eq:gradient-phase-decomposition} times two controlled phases with the rank-one frequency matrices $u'(x')^T/(8g_d)$ and $x'(u')^T/(8g_d)$, so it has Fourier-rank cost at most two, with no rounding or approximation.
\end{remark}

\subsection{From determinant to minimum-eigenvalue estimation}
For the rest of this subsection, let $n$ be even and define the symmetric part of $Q$ by
\[
 S_Q:=\frac{Q+Q^T}{2}.
\]
The nonreal eigenvalues of a real orthogonal matrix come in rotation pairs $e^{\pm i\theta}$.  We represent each pair by an angle $\theta\in[0,\pi]$; these are the positive eigenangles.  On the corresponding real two-dimensional plane, $S_Q$ acts as $\cos\theta$ times the identity.  Thus a rotation pair with eigenangle $\theta$ contributes the eigenvalue $\cos\theta$ to $S_Q$.

Each rotation pair has determinant one. Therefore $\det Q=-1$ forces an eigenvalue $-1$. On the other hand, a Haar-random matrix in $\mathsf{SO}(n)$ usually has no eigenangle very close to $\pi$. We make this separation quantitative below. We use the standard real canonical form and Haar conventions for orthogonal matrices as in \cite{Meckes19}.

The reduction will compare an estimate of $\lambda_{\min}(S_Q)$ with a threshold of the form $-1+g_n$. We first prove that, with high probability, the negative component has minimum eigenvalue $-1$, while the positive component has minimum eigenvalue at least $-1+2g_n$.

We choose the width of the interval near $\pi$ so that its expected number of eigenangles is at most $1/100$. Set
\begin{equation}
 \alpha:=\frac\pi{100},
 \qquad
 g_n:=\frac{1-\cos(\alpha/n)}2.
 \label{eq:constants}
\end{equation}
Since $1-\cos t=\Theta(t^2)$ near zero, we have $g_n=\Theta(n^{-2})$. The definition also gives
\[
 \cos(\pi-\alpha/n)=-\cos(\alpha/n)=-1+2g_n.
\]

\begin{lemma}[Endpoint eigenangle density]
\label{lem:eigenangle-density}
Let $n=2N$. The expected eigenangle density $\rho$ is defined by the identity
\[
 \E\bigl[\textup{number of eigenangles in }J\bigr]
 =\int_J\rho(\theta)\,d\theta
\]
for every measurable $J\subseteq[0,\pi]$.

Under $\mu_{n,+}$, the $N$ positive eigenangles have expected density
\[
 \rho_N^+(\theta)=\frac1\pi+\frac2\pi
 \sum_{j=1}^{N-1}\cos^2(j\theta)\le\frac n\pi.
\]
Under $\mu_{n,-}$, there is almost surely one forced eigenvalue $+1$ and one forced eigenvalue $-1$. After removing them, the $N-1$ positive eigenangles have expected density
\[
 \rho_N^-(\theta)=\frac2\pi
 \sum_{j=1}^{N-1}\sin^2(j\theta)\le\frac n\pi.
\]
Consequently, the expected number of the indicated angles in $J\subseteq[0,\pi]$ is at most $n|J|/\pi$.
\end{lemma}

\begin{proof}
These density formulas follow from the Weyl integration formula for the two even orthogonal components; see~\cite[Chapter~3]{Meckes19}.  The pointwise upper bounds follow term by term.  Integrating over $J$ gives the expectation bound.
\end{proof}

We now define high-probability events on which the two determinant components are separated by a spectral gap of order $n^{-2}$. Let
\[
 J_n:=[\pi-\alpha/n,\pi].
\]
Define $\mathcal E_{n,+}$ to be the event that no positive eigenangle lies in $J_n$. Define $\mathcal E_{n,-}$ to be the event that no nonforced positive eigenangle lies in $J_n$ and that the forced endpoint eigenvalues are simple.

\begin{lemma}[High-probability spectral separation]
\label{lem:regular-events}
For every even $n$,
\begin{equation}
 \mu_{n,+}(\mathcal E_{n,+})\ge\frac{99}{100},
 \qquad
 \mu_{n,-}(\mathcal E_{n,-})\ge\frac{99}{100}.
 \label{eq:regular-probability}
\end{equation}
On the positive event,
\begin{equation}
 Q\in\mathcal E_{n,+}
 \quad\Longrightarrow\quad
 \lambda_{\min}(S_Q)\ge-1+2g_n.
 \label{eq:regular-S-plus}
\end{equation}
On the negative event, let $u$ be a unit $-1$ eigenvector of $S_Q$.  Then
\begin{equation}
 Q\in\mathcal E_{n,-}
 \quad\Longrightarrow\quad
 \lambda_{\min}(S_Q)=-1,\qquad
 \dim\ker(S_Q+I_n)=1,\qquad
 S_Q|_{u^\perp}\succeq(-1+2g_n)I_{u^\perp}.
 \label{eq:regular-S-minus}
\end{equation}
In particular, the eigenvalue $-1$ is simple and every other eigenvalue is at least $-1+2g_n$.
\end{lemma}

\begin{proof}
Lemma~\ref{lem:eigenangle-density} bounds the expected number of relevant angles in $J_n$ by
\[
 \frac{n|J_n|}{\pi}=\frac{n(\alpha/n)}\pi=\frac1{100}.
\]
Markov's inequality gives \eqref{eq:regular-probability}.  Endpoint multiplicities larger than the forced ones form a Haar-null set.

If $\theta\le\pi-\alpha/n$, then
\[
 \cos\theta\ge\cos(\pi-\alpha/n)=-1+2g_n.
\]
On $\mathcal E_{n,+}$, this applies to every eigenangle and proves \eqref{eq:regular-S-plus}. On $\mathcal E_{n,-}$, the forced eigenvalue $-1$ is simple, and the same bound applies to every other eigenvalue. This proves \eqref{eq:regular-S-minus}.
\end{proof}

One phase query to $S_Q$ can be implemented with two phase queries to $Q$.  Indeed,
\begin{equation}
 \begin{aligned}
  u^TS_Qv
  &=\frac12u^TQv+\frac12v^TQu,\\
  e^{iu^TS_Qv}
  &=e^{i(u/2)^TQv}\,e^{i(v/2)^TQu}.
 \end{aligned}
 \label{eq:S-query-simulation}
\end{equation}
The first factor is one $Q$ phase query with a rescaled Fourier label. The second swaps the roles of the two query labels and uses one more $Q$ phase query. This transpose-access equivalence is standard in the continuous matrix phase-query model~\cite{childs_et_al:LIPIcs.ICALP.2021.55}.
The same construction with the signs reversed simulates an inverse phase query to $S_Q$.

\begin{theorem}[High-accuracy minimum-eigenvalue lower bound]
\label{thm:min-eigenvalue}
Let $n$ be even. Any algorithm with matrix phase-oracle access to $S_Q$ which, for every $Q\in \mathsf{O}(n)$, estimates $\lambda_{\min}(S_Q)$ to additive error at most $g_n/3$ with success probability at least $2/3$ makes at least $n/4$ queries. In particular, the task requires $\Omega(n)$ queries at accuracy $\Theta(n^{-2})$.
\end{theorem}

\begin{proof}
Suppose the estimator makes $T$ queries to $S_Q$.  Given its estimate $\widetilde\lambda$, classify the determinant by the threshold
\[
 \tau:=-1+g_n.
\]
Output $+1$ when $\widetilde\lambda>\tau$ and output $-1$ otherwise.

If $Q\in\mathcal E_{n,-}$ and the estimate succeeds, \eqref{eq:regular-S-minus} gives
\[
 \widetilde\lambda\le-1+\frac{g_n}{3}<-1+g_n.
\]
The classifier therefore outputs $-1$.

If $Q\in\mathcal E_{n,+}$ and the estimate succeeds, \eqref{eq:regular-S-plus} gives
\[
 \widetilde\lambda\ge-1+2g_n-\frac{g_n}{3}
 =-1+\frac{5g_n}{3}>-1+g_n.
\]
The classifier therefore outputs $+1$.

On each determinant component, the corresponding event $\mathcal E_{n,\pm}$ has probability at least $99/100$. At every input, the estimator succeeds with probability at least $2/3$. Hence the classifier is correct on each component with probability at least
\[
 \frac{99}{100}\cdot\frac23=\frac{33}{50}>\frac12.
\]
Its Haar-average success probability is therefore strictly greater than $1/2$.

By \eqref{eq:S-query-simulation}, the $T$ queries to $S_Q$ use $2T$ queries to $Q$. The determinant lower bound in Theorem~\ref{thm:mv-determinant} gives $2T\ge n/2$. Thus $T\ge n/4$.
\end{proof}

Since every $S_Q$ is a real symmetric matrix, the theorem gives a worst-case lower bound for real symmetric matrices in every even dimension. For odd $n$, apply the construction in dimension $n-1$ and append one known eigenvalue:
\[
 \widetilde S_Q:=S_Q\oplus[1].
\]
Then $\lambda_{\min}(\widetilde S_Q)=\lambda_{\min}(S_Q)$. A phase query to $\widetilde S_Q$ is a phase query to $S_Q$ multiplied by the known phase $e^{iu_nv_n}$, so it uses one query to $S_Q$. Thus minimum-eigenvalue estimation for real symmetric $n\times n$ matrices requires $\Omega(n)$ matrix phase queries at accuracy $\Theta(n^{-2})$ in every dimension.

\section{From determinant hardness to convex optimization}
\label{sec:determinant-to-optimization}

The hard family in this section is indexed by a hidden orthogonal matrix. Fix the ambient optimization dimension $n$, set
\[
 d:=2\lfloor n/2\rfloor
\]
and let $Q\in \mathsf{O}(d)$ be the hidden matrix. Thus $d=n$ when $n$ is even and $d=n-1$ when $n$ is odd; in the latter case, we append one fixed coordinate to the construction. We use the notation of Section~\ref{sec:matrix-lower-bounds}: in particular, the lower bound analyzes $Q$ drawn from the Haar measure on $\mathsf{O}(d)$, where
\[
 S_Q=\frac{Q+Q^T}{2},
\]
the spectral scale is denoted $g_d$, and the regular events in the two determinant components are $\mathcal E_{d,+}$ and $\mathcal E_{d,-}$. 

We now outline the hard family of convex sets. We shift and rescale $S_Q$ to obtain a positive-definite matrix $A_Q$, and use $A_Q$ as the quadratic form of a centered ellipsoid, up to an overall scale. The bottom of the spectrum of $A_Q$ preserves the information of the sign of $\det Q$. Whenever $\det Q=-1$, the matrix $A_Q$ has the eigenvalue $2g_d$. On $\mathcal E_{d,-}$, this eigenvalue is simple and all other eigenvalues are at least $3g_d$. On $\mathcal E_{d,+}$, every eigenvalue is at least $3g_d$.

This family has two useful features. First, membership of a represented point can be written as a threshold
\[
 q^TQq \le \tau.
\]
The Boolean phase associated with one membership query therefore has Fourier rank at most one (Lemma \ref{lem:threshold-rank-one}), so predicting $\det Q$ still requires at least $d/2$ membership queries. Second, the maximum of a linear objective $b^Tx$ over the ellipsoid $\{x:x^TA_Qx\le1\}$ equals $(b^TA_Q^{-1}b)^{1/2}$. This maximum is large when $b$ has a substantial component along an eigenvector of $A_Q$ with a small eigenvalue. For a random sign vector $b$, the maximum is therefore more likely to be large when $\det Q=-1$. A test that compares an approximate maximum with a fixed threshold distinguishes the two determinant components with probability strictly greater than $1/2$.

Suppose now that an optimizer uses at most $T$ membership queries per invocation. The determinant predictor in Section~\ref{sec:random-objective} uses $125$ optimizer invocations with the same random objective. The membership-query lower bound for the determinant then forces
\[
 125T\ge\frac d2.
\]
The first version of this paper used a different reduction, based on inverse iteration, which needs $O(\log d\,\log\log d)$ optimizer invocations. Section~\ref{sec:inverse-iteration-argument} presents that reduction.

The next subsection gives the exact matrices and ellipsoids. The second subsection proves the membership lower bound. The last two subsections construct the determinant predictor from one random objective.

\subsection{The hard family of ellipsoids}
\label{sec:hard-family}

Section~\ref{sec:matrix-lower-bounds} established a difference between the two determinant components at the bottom of the spectrum of $S_Q$. On $\mathcal E_{d,-}$, the matrix $S_Q$ has one simple eigenvalue $-1$, while every other eigenvalue is at least $-1+2g_d$. On $\mathcal E_{d,+}$, every eigenvalue is at least $-1+2g_d$. We encode this difference in a positive-definite matrix. First, the transformation $H_Q=(I_d-S_Q)/2$ maps the spectrum of $S_Q$ from $[-1,1]$ to $[0,1]$. We then reverse the order of these eigenvalues and shift them upward by $2g_d$. Define
\begin{equation}
 H_Q:=\frac{I_d-S_Q}{2}
     =\frac12I_d-\frac14(Q+Q^T),
 \qquad
 A_Q:=(1+2g_d)I_d-H_Q.
 \label{eq:hard-matrices}
\end{equation}
Thus
\begin{equation}
 A_Q=\left(\frac12+2g_d\right)I_d+\frac14(Q+Q^T).
 \label{eq:A-explicit}
\end{equation}
On $\mathcal E_{d,-}$, $A_Q$ has one simple eigenvalue equal to $2g_d$. On $\mathcal E_{d,+}$, all its eigenvalues are at least $3g_d$. This is the gap used by inverse iteration in Section~\ref{sec:inverse-iteration-argument}.

\begin{lemma}[Spectral geometry]
\label{lem:spectral-geometry}
For every $Q\in \mathsf{O}(d)$,
\begin{equation}
 0\preceq H_Q\preceq I_d,
 \qquad
 2g_dI_d\preceq A_Q\preceq(1+2g_d)I_d.
 \label{eq:spectral-bounds}
\end{equation}
A rotation pair with eigenangle $\theta\in[0,\pi]$ contributes the eigenvalue $(1-\cos\theta)/2$ to $H_Q$.

On the positive regular event,
\begin{equation}
 Q\in\mathcal E_{d,+}
 \quad\Longrightarrow\quad
 H_Q\preceq(1-g_d)I_d,
 \qquad
 A_Q\succeq3g_dI_d.
 \label{eq:regular-plus}
\end{equation}
On the negative regular event, let $u$ be the unit vector from \eqref{eq:regular-S-minus}. The matrix $H_Q$ has simple top eigenvalue $1$, and all its other eigenvalues are at most $1-g_d$:
\begin{equation}
 H_Qu=u,
 \qquad
 \dim\ker(I_d-H_Q)=1,
 \qquad
 H_Q|_{u^\perp}\preceq(1-g_d)I_{u^\perp}.
 \label{eq:regular-minus}
\end{equation}
Consequently, $A_Q$ has a simple eigenvalue $2g_d$ and satisfies
\begin{equation}
 A_Qu=2g_du,
 \qquad
 A_Q|_{u^\perp}\succeq3g_dI_{u^\perp}.
 \label{eq:inverse-gap}
\end{equation}
\end{lemma}

\begin{proof}
For real $z$, the scalar identity $z^TQ^Tz=z^TQz$ and orthogonality give $\abs{z^TQz}\le\norm z^2$. Hence
\[
 z^TH_Qz=\frac12\norm z^2-\frac12z^TQz
 \in[0,\norm z^2].
\]
This proves $0\preceq H_Q\preceq I_d$. Subtracting these bounds from $(1+2g_d)I_d$ gives the two bounds on $A_Q$ in \eqref{eq:spectral-bounds}.

On a real rotation block of angle $\theta$, we have $S_Q=\cos\theta I_2$. Therefore $H_Q=(I_d-S_Q)/2$ has eigenvalue $(1-\cos\theta)/2$ on that block. This is the standard real canonical form used in~\cite{Meckes19}.

If $Q\in\mathcal E_{d,+}$, equation~\eqref{eq:regular-S-plus} and the identity $H_Q=(I_d-S_Q)/2$ give $H_Q\preceq(1-g_d)I_d$. The definition of $A_Q$ then gives $A_Q\succeq3g_dI_d$. This proves \eqref{eq:regular-plus}.

If $Q\in\mathcal E_{d,-}$, equation~\eqref{eq:regular-S-minus} gives $S_Qu=-u$, with this eigenvalue simple, and $S_Q|_{u^\perp}\succeq(-1+2g_d)I_{u^\perp}$. Applying $H_Q=(I_d-S_Q)/2$ gives \eqref{eq:regular-minus}. Subtracting these eigenvalues from $1+2g_d$ gives \eqref{eq:inverse-gap} and proves that the eigenvalue $2g_d$ of $A_Q$ is simple.
\end{proof}

To append the fixed coordinate when $n$ is odd, define
\begin{equation}
 d_n:=d,
 \qquad
 \ell_n:=n-d_n\in\{0,1\},
 \qquad
 \overline A_{n,Q}:=
 \begin{cases}
 A_Q,&\ell_n=0,\\
 A_Q\oplus[1],&\ell_n=1.
 \end{cases}
 \label{eq:odd-embedding}
\end{equation}
For use after this section, we package the hidden-block parameters in notation indexed only by the ambient dimension:
\begin{equation}
 \mathcal Q_n:=\mathsf{O}(d_n),
 \qquad
 \overline g_n:=g_{d_n},
 \qquad
 \overline{\mathcal E}_{n,\pm}:=\mathcal E_{d_n,\pm}.
 \label{eq:ambient-notation}
\end{equation}
For a fixed scale $s>0$, define the ellipsoid
\begin{equation}
 K_{s,Q}:=\{y\in\R^n:y^T\overline A_{n,Q}y\le s^2\}.
 \label{eq:scaled-body}
\end{equation}
The spectral bounds in \eqref{eq:spectral-bounds} give
\begin{equation}
 B_2(0,r_s)\subseteq K_{s,Q}\subseteq B_2(0,R_s),
 \qquad
 r_s:=\frac{s}{\sqrt{1+2g_d}},
 \qquad
 R_s:=\frac{s}{\sqrt{2g_d}},
 \qquad d=d_n.
 \label{eq:scaled-ball-bounds}
\end{equation}
These are the hard ellipsoids used throughout the rest of the proof.

\subsection{Determinant hardness for membership queries}
\label{sec:membership-rank}

We begin with the hard problem that will be reduced to optimization. For the family $K_{s,Q}$, a membership answer is a threshold function of $q^TQq$. Its Fourier frequencies are scalar multiples of $qq^T$ and therefore have rank at most one. It follows that $M$ membership queries can produce output probabilities of Fourier rank at most $2M$. Reflection orthogonality from Lemma~\ref{lem:reflection} then rules out any correlation with $\det Q$ when $2M<d$.

\paragraph{Discretized query points.}
The finite grid makes the membership register finite, allowing us to analyze the membership phase separately at each possible query point. Its spacing does not enter the Fourier-rank argument, so the grid may be chosen as finely as desired. The oracle returns the exact membership bit at every represented point.

Fix $d=d_n$, a scale $s>0$, and a finite grid $\mathcal G\subseteq\R^n$ chosen independently of $Q$. By \eqref{eq:scaled-ball-bounds}, $K_{s,Q}\subseteq B_2(0,R_s)$; hence only the finite set
\[
 \mathcal X:=\mathcal G\cap B_2(0,R_s)
\]
has a $Q$-dependent membership answer.

For $y=(y',y'')\in\R^d\times\R^{\ell_n}$, define
\begin{equation}
 q_y:=\frac{y'}s,
 \qquad
 t_y:=\frac{\norm{y''}}s,
 \qquad
 \tau_y:=2-(1+4g_d)\norm{q_y}^2-2t_y^2.
 \label{eq:membership-parameters}
\end{equation}
We now rewrite the membership predicate in a form that exposes its dependence on $Q$. Equation~\eqref{eq:A-explicit} gives
\[
 \frac{y^T\overline A_{n,Q}y}{s^2}
 =\left(\frac12+2g_d\right)\norm{q_y}^2
    +\frac12q_y^TQq_y+t_y^2.
\]
Therefore
\begin{equation}
 y\in K_{s,Q}
 \quad\Longleftrightarrow\quad
 q_y^TQq_y\le\tau_y.
 \label{eq:membership-threshold}
\end{equation}

Our next goal is to quantify how much Fourier rank one membership answer can introduce. For a fixed query point, all dependence on the hidden matrix is through the single scalar $q^TQq$. Although thresholding this scalar gives a discontinuous function, the threshold can be approximated in $L^2$ by trigonometric polynomials in $q^TQq$. Each exponential in such a polynomial has the form
\[
 e^{itq^TQq}=\chi_{tqq^T}(Q),
\]
whose frequency matrix $tqq^T$ has rank at most one. The lemma formalizes this observation.

\begin{lemma}[The membership phase has Fourier rank at most one]
\label{lem:threshold-rank-one}
Fix $q\in\R^d$ and $\tau\in\R$. The corresponding membership phase is the following sign-valued function of the hidden matrix $Q$:
\[
 \phi_{q,\tau}(Q):=(-1)^{\1[q^TQq\le\tau]}.
\]
Thus $\phi_{q,\tau}(Q)=-1$ when the threshold test returns $1$, and $\phi_{q,\tau}(Q)=+1$ when it returns $0$. For the membership query at a represented point $y$, the relevant parameters are $q=q_y$ and $\tau=\tau_y$ from \eqref{eq:membership-parameters}.
Then $\phi_{q,\tau}\in\mathcal V_1^{(d)}$.  More precisely, there are trigonometric polynomials
\begin{equation}
 p_a(Q)=\sum_{j=-a}^a c_{a,j}e^{it_{a,j}q^TQq}
       =\sum_{j=-a}^a c_{a,j}\chi_{t_{a,j}qq^T}(Q)
 \label{eq:Fejer-approximation}
\end{equation}
with $\abs{p_a(Q)}\le1$ such that $p_a\to\phi_{q,\tau}$ almost everywhere and in $L^2(\mathsf{O}(d),\mu_d)$.
\end{lemma}

\begin{proof}
The case $q=0$ is constant. Otherwise, $q^TQq/\norm q^2$ has the distribution of one coordinate of a uniform point on the sphere and is atomless. On the interval $[-\norm q^2,\norm q^2]$, let $h(t)=(-1)^{\1[t\le\tau]}$ and extend this function periodically. The resulting periodic function has only finitely many jumps per period. Its Fej\'er means from Section~\ref{sec:Fejer-prelim} are bounded in modulus by one and converge at every continuity point. Composing with $q^TQq$ gives almost-everywhere convergence, while every Fourier term has frequency $tqq^T$ of rank at most one. Dominated convergence supplies the $L^2$ conclusion.
\end{proof}

Note that “Fourier rank at most one” does not mean that $\phi_{q,\tau}$ is a single exponential. It means that this exact Boolean phase lies in the $L^2$-closed span of exponentials whose frequency matrices have rank at most one. The Fej\'er polynomials are used to establish this analytic statement, they are not implemented as actual queries.

In the Hadamard basis of the answer qubit, a membership query acts as multiplication by the membership phase via phase-kickback. The preceding lemma therefore gives the same rank-growth rule as a matrix phase query.

\begin{lemma}[Fourier-rank growth for exact membership]
\label{lem:membership-rank-growth}
After $M$ calls to exact membership for $K_{s,Q}$, every purified algorithmic state lies in $\mathcal V_M^{(d)}\widehat\otimes\mathcal H_{\rm alg}$. Every final outcome probability belongs to $\mathcal A_{2M}^{(d)}$.
\end{lemma}

\begin{proof}
Use the Hadamard basis $\ket{\widehat z}=2^{-1/2}\sum_{b=0}^1(-1)^{zb}\ket b$ of the answer qubit. For a query point $y$, one membership call fixes the $z=0$ component and multiplies the $z=1$ component by $\phi_{q_y,\tau_y}$.  If $\Psi\in\mathcal V_r^{(d)}\widehat\otimes\mathcal H_{\rm alg}$, use the bounded approximants from Lemma~\ref{lem:threshold-rank-one}.  Multiplication by each $p_a$ adds a rank-one frequency, and
\[
 \norm{(p_a-\phi_{q_y,\tau_y})\Psi}_{L^2}^2
 =\int\abs{p_a-\phi_{q_y,\tau_y}}^2\norm{\Psi(Q)}^2d\mu_d(Q)
 \longrightarrow0
\]
by dominated convergence. Closedness therefore shows that one controlled query maps rank $r$ to rank at most $r+1$. There are finitely many represented query points, so the complete membership operation is a finite controlled sum of these maps. The membership oracle is an involution, so the same bound holds for its inverse. Theorem~\ref{thm:Fourier-rank-method} now proves both statements.
\end{proof}

\begin{theorem}[Determinant lower bound for exact membership]
\label{thm:binary-determinant}
Any algorithm which predicts $\det Q$ for $Q\sim\mu_d$ with average success probability greater than $1/2$, using only exact membership for $K_{s,Q}$ on a finite grid fixed in advance, makes at least $d/2$ queries.
\end{theorem}

\begin{proof}
If $M<d/2$, Lemma~\ref{lem:membership-rank-growth} places the output-$+1$ probability in $\mathcal A_{2M}^{(d)}\subseteq\mathcal A_{d-1}^{(d)}$. Equation~\eqref{eq:reflection-L1} makes its determinant correlation zero. The decision conclusion of Theorem~\ref{thm:Fourier-rank-method} then makes the average success exactly $1/2$.
\end{proof}

\paragraph{Continuous query registers.}
Theorem~\ref{thm:binary-determinant} is the finite-register result used in the main optimization lower bounds. For completeness, we also consider the idealized model in which the membership query point ranges over a measurable space $Y$ and the query register is $L^2(Y,\nu)$. The finite controlled-sum argument in Lemma~\ref{lem:membership-rank-growth} does not apply directly to an uncountable set of query points. Nevertheless, the same rank-one argument goes through by constructing jointly measurable Fej\'er approximants and taking a strong $L^2$ limit. The following corollary records this extension. It is used only for the continuous-register version of the optimization result stated after the main theorems; the proof is given in Appendix~\ref{app:continuous-membership}.

\begin{corollary}[Continuous exact-membership determinant lower bound]
\label{cor:continuous-membership-determinant}
Let $(Y,\nu)$ be a sigma-finite query space, and let $q:Y\to\R^d$ and $\tau:Y\to\R$ be measurable. For each $Q\in \mathsf{O}(d)$, define the Boolean threshold predicate
\[
 m_Q(y):=\1[q(y)^TQq(y)\le\tau(y)].
\]
Any algorithm which predicts $\det Q$ for $Q\sim\mu_d$ with average success probability greater than $1/2$, using exact coherent query access to $m_Q$ on $L^2(Y,\nu)\otimes\C^2$, makes at least $d/2$ queries.
\end{corollary}

\subsection{A random objective separates the determinant components}
\label{sec:random-objective}

We now show that the optimal value of one random linear objective carries information about $\det Q$. Throughout this subsection, $d\ge2$ is even and $Q\in\mathsf O(d)$. Recall $\alpha$ and $g_d$ from \eqref{eq:constants} and the matrix $A_Q$ from \eqref{eq:hard-matrices}. Put
\begin{equation}
 a_d:=2g_d=1-\cos(\alpha/d).
 \label{eq:direct-scale}
\end{equation}
Then \eqref{eq:A-explicit} and the inequality $1-\cos t\le t^2/2$ give
\begin{equation}
 A_Q=a_dI_d+\frac{I_d+S_Q}{2},
 \qquad
 a_dI_d\preceq A_Q\preceq(1+a_d)I_d,
 \qquad
 a_d\le\frac{\alpha^2}{2d^2}.
 \label{eq:direct-A}
\end{equation}

For a unit vector $v$, the maximum of $v^Tx$ over the ellipsoid $\{x:x^TA_Qx\le1\}$ is $(v^TA_Q^{-1}v)^{1/2}$. This quantity is large when $v$ has a substantial component along an eigenvector of $A_Q$ with a small eigenvalue. If $\det Q=-1$, then $Q$ has the eigenvalue $-1$, and $A_Q$ has the eigenvalue $a_d=\Theta(d^{-2})$ on the corresponding eigenvector. A random sign vector has squared overlap of order $1/d$ with this eigenvector with constant probability, so $v^TA_Q^{-1}v$ is at least of order $d$ with constant probability. If $\det Q=+1$, the matrix $A_Q$ may also have small eigenvalues, but Lemma~\ref{lem:eigenangle-density} shows that few eigenangles of $Q$ lie near $\pi$. The average of $v^TA_Q^{-1}v$ over random sign vectors equals $\Tr(A_Q^{-1})/d$, and its expectation over $Q$ is also of order $d$.

The two determinant components are therefore not separated by an order of magnitude. They are separated by constants, and the small value of $\alpha$ in \eqref{eq:constants} makes the constant on the positive component small. The resulting predictor has only a small advantage over random guessing. Any positive advantage suffices, however, because Theorem~\ref{thm:binary-determinant} shows that fewer than $d/2$ membership queries give success probability exactly $1/2$.

We now make this precise. Choose
\[
 v=d^{-1/2}(\sigma_1,\ldots,\sigma_d),
\]
where $\sigma_1,\ldots,\sigma_d$ are independent and uniform on $\{-1,+1\}$ and are chosen independently of $Q$. Define
\begin{equation}
 Z(Q,v):=\sqrt{a_d d\,v^TA_Q^{-1}v}.
 \label{eq:direct-statistic}
\end{equation}
The factor $a_dd$ normalizes $Z$ so that the thresholds below are absolute constants. We work with $Z$ rather than $Z^2$ because an optimizer returns the maximum value $(v^TA_Q^{-1}v)^{1/2}$ up to an additive error, which becomes an additive error in $Z$. Let $\Prb_\pm$ and $\E_\pm$ denote probability and expectation with respect to $Q\sim\mu_{d,\pm}$ and the independent random vector $v$.

The analysis of the negative component uses two moments of the overlap between $v$ and a fixed unit vector.

\begin{lemma}[Rademacher moments]
\label{lem:rademacher-moments}
Let $u\in\R^d$ be a unit vector, let $v$ be the random sign vector above, and put $X:=d\abs{u^Tv}^2$. Then
\[
 \E X=1,
 \qquad
 \E X^2=3-2\sum_{i=1}^du_i^4\le3.
\]
\end{lemma}

\begin{proof}
We have $X=(\sum_iu_i\sigma_i)^2$. Independence and $\E\sigma_i\sigma_j=\1[i=j]$ give $\E X=\sum_iu_i^2=1$. In the expansion of $(\sum_iu_i\sigma_i)^4$, a term has nonzero expectation only if every index appears an even number of times. Hence
\[
 \E X^2=\sum_iu_i^4+3\sum_{i\ne j}u_i^2u_j^2
 =3\Big(\sum_iu_i^2\Big)^2-2\sum_iu_i^4
 =3-2\sum_iu_i^4.
\]
\end{proof}

\begin{lemma}[Separation of the determinant components]
\label{lem:direct-separation}
The statistic in \eqref{eq:direct-statistic} satisfies
\begin{equation}
 \Prb_-\!\left[Z\ge\frac1{\sqrt3}\right]\ge\frac4{27},
 \qquad
 \E_+[Z^2]\le\frac{2d\sqrt{a_d}}{\sqrt{1+a_d}}\le\sqrt2\alpha<\frac9{200}.
 \label{eq:direct-separation}
\end{equation}
\end{lemma}

\begin{proof}
We first treat the negative component. Fix $Q\in\mathsf O^-(d)$. Every real orthogonal matrix with negative determinant has a unit eigenvector $u$ with eigenvalue $-1$. Then $Q^Tu=-u$, so $S_Qu=-u$, and \eqref{eq:direct-A} gives $A_Qu=a_du$. Since $A_Q$ is positive definite, $u$ is also an eigenvector of $A_Q^{-1}$ with eigenvalue $a_d^{-1}$, and the remaining eigenvalues of $A_Q^{-1}$ are positive. Hence $v^TA_Q^{-1}v\ge a_d^{-1}\abs{u^Tv}^2$, and
\[
 Z^2\ge d\abs{u^Tv}^2.
\]
Let $X:=d\abs{u^Tv}^2$. Lemma~\ref{lem:rademacher-moments} gives $\E_v[X]=1$ and $\E_v[X^2]\le3$. The Paley--Zygmund inequality~\cite{paley1932some} with parameter $1/3$ therefore yields
\[
 \Prb_v[X\ge1/3]
 \ge\left(1-\frac13\right)^2\frac{\E_v[X]^2}{\E_v[X^2]}
 \ge\frac4{27}.
\]
On this event, $Z\ge1/\sqrt3$. The bound holds for every $Q\in\mathsf O^-(d)$, so it also holds after averaging over $Q\sim\mu_{d,-}$. This part of the argument does not use the event $\mathcal E_{d,-}$ and does not require the eigenvalue $-1$ to be simple.

We next treat the positive component. Since $\E_v[vv^T]=I_d/d$,
\begin{equation}
 \E_+[Z^2]
 =a_d\E_{Q\sim\mu_{d,+}}\Tr(A_Q^{-1}).
 \label{eq:direct-trace-average}
\end{equation}
Let $\rho_{d/2}^+$ be the expected positive-eigenangle density from Lemma~\ref{lem:eigenangle-density}. Almost surely, the eigenvalues of $Q\in\mathsf{SO}(d)$ consist of $d/2$ rotation pairs. By \eqref{eq:direct-A}, each positive eigenangle $\theta\in[0,\pi]$ contributes two eigenvalues
\[
 a_d+\frac{1+\cos\theta}{2}
\]
to $A_Q$. Since $\rho_{d/2}^+(\theta)\le d/\pi$, equation~\eqref{eq:direct-trace-average} gives
\begin{align}
 \E_+[Z^2]
 &=2a_d\int_0^\pi\frac{\rho_{d/2}^+(\theta)}{a_d+(1+\cos\theta)/2}\,d\theta \notag\\
 &\le\frac{2a_d d}{\pi}\int_0^\pi\frac{d\theta}{a_d+(1+\cos\theta)/2} \notag\\
 &=\frac{2d\sqrt{a_d}}{\sqrt{1+a_d}}.
 \label{eq:direct-positive-moment}
\end{align}
The last identity follows from the standard integral
\[
 \int_0^\pi\frac{d\theta}{\omega_0+\omega_1\cos\theta}
 =\frac{\pi}{\sqrt{\omega_0^2-\omega_1^2}},
 \qquad \omega_0>\abs{\omega_1},
\]
with $\omega_0=a_d+1/2$ and $\omega_1=1/2$, for which $\omega_0^2-\omega_1^2=a_d(1+a_d)$. Finally, \eqref{eq:direct-A} gives $d\sqrt{a_d}\le\alpha/\sqrt2$, so the right side of \eqref{eq:direct-positive-moment} is at most $\sqrt2\alpha=\sqrt2\pi/100<9/200$.
\end{proof}

\subsection{A determinant predictor from one random objective}
\label{sec:direct-predictor}

Lemma~\ref{lem:direct-separation} shows that $Z$ tends to be larger on the negative component. We now turn an additive estimate of $Z$ into a determinant predictor. Suppose a procedure, given the hidden matrix through some oracle and a fixed vector $v$, returns a real number $\widetilde Z$ that is close to $Z(Q,v)$ with probability at least $2/3$. Consider the following predictor:
\begin{enumerate}[label=\textup{(\arabic*)}]
 \item sample the random sign vector $v$;
 \item run the procedure independently $m:=125$ times with this vector $v$;
 \item let $\widehat Z$ be the median of the $125$ outputs;
 \item output $-1$ if $\widehat Z>9/16$, and output $+1$ otherwise.
\end{enumerate}
The threshold $9/16$ must satisfy two conditions. On the negative component, a median within $1/200$ of a value $Z\ge1/\sqrt3$ must exceed the threshold, which holds because $9/16+1/200<1/\sqrt3$. On the positive component, Markov's inequality must bound the probability that $Z$ exceeds the threshold minus $1/200$ by a number smaller than $4/27$, which the calculation below verifies.

\begin{lemma}[Determinant prediction from a scalar estimate]
\label{lem:direct-classifier}
Suppose that, for every $Q\in\mathsf O(d)$ and every sign vector $v$, the procedure satisfies
\[
 \Prb\!\left[\abs{\widetilde Z-Z(Q,v)}\le\frac1{200}\right]\ge\frac23.
\]
Then the predictor above succeeds with probability greater than $1/2+1/2000$, averaged over $Q\sim\mu_d$, the vector $v$, and the internal randomness of the procedure.
\end{lemma}

\begin{proof}
Fix $(Q,v)$ and call an execution successful if its output lies within $1/200$ of $Z(Q,v)$. The $125$ executions are independent, so Hoeffding's inequality~\cite{hoeffding1963probability} bounds the probability that at most half of them succeed by
\[
 \exp(-125/18)<\delta_0:=\frac1{1000}.
\]
Whenever a strict majority succeeds, the median lies within $1/200$ of $Z(Q,v)$, regardless of the remaining outputs. This conditional bound holds for every $(Q,v)$, so it also holds after averaging over $Q$ and $v$.

On the negative component, Lemma~\ref{lem:direct-separation} gives $Z\ge1/\sqrt3$ with probability at least $4/27$. On this event, a strict successful majority gives $\widehat Z\ge1/\sqrt3-1/200>9/16$, so the predictor outputs $-1$. Hence
\begin{equation}
 \Prb_-[\textup{output }-1]\ge\frac4{27}-\delta_0.
 \label{eq:direct-negative-classification}
\end{equation}

On the positive component, suppose that a strict majority succeeds and the predictor outputs $-1$. Then $Z>9/16-1/200=223/400$. Markov's inequality applied to $Z^2$ and Lemma~\ref{lem:direct-separation} give
\begin{equation}
 \Prb_+[\textup{output }-1]
 \le\frac{9/200}{(223/400)^2}+\delta_0
 =\frac{7200}{49729}+\delta_0.
 \label{eq:direct-positive-classification}
\end{equation}

The two determinant components have equal Haar measure. The average success probability is therefore
\[
 \frac12+\frac12\Bigl(\Prb_-[\textup{output }-1]-\Prb_+[\textup{output }-1]\Bigr).
\]
By \eqref{eq:direct-negative-classification} and \eqref{eq:direct-positive-classification}, the difference in parentheses is at least
\[
 \frac4{27}-\frac{7200}{49729}-\frac2{1000}
 =\frac{4516}{1342683}-\frac2{1000}
 >\frac3{1000}-\frac2{1000}
 =\frac1{1000}.
\]
The average success probability is therefore greater than $1/2+1/2000$.
\end{proof}

\section{Main result: optimization lower bounds}
\label{sec:optimization-lower-bounds}

We now combine the determinant predictor of Section~\ref{sec:direct-predictor} with the membership-query determinant lower bound. Section~\ref{sec:linear-lower-bound} proves an $\Omega(n)$ lower bound for every scale of the hard ellipsoids and states it for ellipsoids contained in the unit ball.

\subsection{Linear membership-query lower bounds}
\label{sec:linear-lower-bound}

Fix an ambient dimension $n\ge2$, let $d=d_n=2\lfloor n/2\rfloor$, and let $\ell_n=n-d$ as in \eqref{eq:odd-embedding}. We use the notation $a_d=2g_d$ from \eqref{eq:direct-scale}. For a scale $s>0$, equation~\eqref{eq:scaled-ball-bounds} gives the inclusions
\[
 B_2(0,r_s)\subseteq K_{s,Q}\subseteq B_2(0,R_s),
 \qquad
 r_s=\frac{s}{\sqrt{1+a_d}},
 \qquad
 R_s=\frac{s}{\sqrt{a_d}},
\]
for every $Q\in\mathcal Q_n$. The ratio of the two radii is
\begin{equation}
 \frac{R_s}{r_s}=\sqrt{\frac{1+a_d}{a_d}},
 \label{eq:direct-aspect-ratio}
\end{equation}
which does not depend on $s$. For a unit vector $v\in\R^d$, put $\overline v=(v,0_{\ell_n})$. Since $\overline A_{n,Q}^{-1}\overline v=(A_Q^{-1}v,0_{\ell_n})$, the definition \eqref{eq:scaled-body} gives
\begin{equation}
 h_{K_{s,Q}}(\overline v)
 =s\sqrt{v^TA_Q^{-1}v},
 \qquad
 Z(Q,v)=\frac{\sqrt{a_d d}}{s}h_{K_{s,Q}}(\overline v).
 \label{eq:direct-support-value}
\end{equation}
An estimate of the maximum value $h_{K_{s,Q}}(\overline v)$ with additive error $s/(200\sqrt{a_dd})$ is therefore an estimate of $Z(Q,v)$ with additive error $1/200$.

The reduction in Section~\ref{sec:direct-predictor} uses only objectives of the form $\overline v$ with $v$ a normalized sign vector. The theorem below therefore assumes success only on these objectives. In particular, it applies to every optimizer that succeeds for all unit objectives.

\begin{theorem}[Linear membership-query lower bounds]
\label{thm:direct-membership-lower-bound}
Let $n\ge2$ and $s>0$. Suppose that, for every $Q\in\mathcal Q_n$ and every objective $\overline v=(d^{-1/2}\sigma,0_{\ell_n})$ with $\sigma\in\{-1,+1\}^d$, an optimizer succeeds with probability at least $2/3$ and uses at most $T$ exact membership queries to $K_{s,Q}$ per invocation. Assume either of the following conditions:
\begin{enumerate}[label=\textup{(\roman*)}]
 \item the optimizer is exactly feasible in the sense of Definition~\ref{def:exactly-feasible-opt} and has accuracy
 \begin{equation}
  \eps\le\frac{s}{200\sqrt{a_d d}};
  \label{eq:direct-exact-accuracy}
 \end{equation}
 \item the optimizer is approximately feasible in the sense of Definition~\ref{def:approximately-feasible-opt} and has tolerance
 \begin{equation}
  \delta_{\rm app}\le\frac{s}{200\sqrt{a_d d}\,(R_s/r_s+2)}.
  \label{eq:direct-approximate-accuracy}
 \end{equation}
\end{enumerate}
Then $T\ge d/250\ge n/500$. The lower bound holds on every nonempty finite grid fixed independently of $Q$. It also holds for exact coherent membership on $L^2(\R^n)\otimes\C^2$.
\end{theorem}

\begin{proof}
Fix $Q$ and a normalized sign vector $v$. We first show that one optimizer invocation with objective $\overline v$ gives an estimate $\widetilde Z$ of $Z(Q,v)$ that satisfies the hypothesis of Lemma~\ref{lem:direct-classifier}. Let $y$ be the output and put
\[
 \widetilde Z:=\frac{\sqrt{a_d d}}{s}\overline v^Ty.
\]

Suppose first that the optimizer is exactly feasible. On every successful invocation,
\[
 h_{K_{s,Q}}(\overline v)-\eps
 \le\overline v^Ty
 \le h_{K_{s,Q}}(\overline v).
\]
Equations~\eqref{eq:direct-support-value} and~\eqref{eq:direct-exact-accuracy} show that $\abs{\widetilde Z-Z(Q,v)}\le1/200$.

Now suppose that the optimizer is approximately feasible. Its tolerance satisfies the requirement $\delta_{\rm app}<r_s$ in Definition~\ref{def:approximately-feasible-opt}, because \eqref{eq:direct-aspect-ratio} gives $\sqrt{a_d}=s/R_s$ and therefore
\[
 \frac{\delta_{\rm app}}{r_s}
 \le\frac{R_s/r_s}{200\sqrt d\,(R_s/r_s+2)}
 <1.
\]
Lemma~\ref{lem:approximately-feasible-geometry}, applied with $r=r_s$ and $R=R_s$, shows that every successful output satisfies
\[
 h_{K_{s,Q}}(\overline v)-D
 \le\overline v^Ty
 \le h_{K_{s,Q}}(\overline v)+\delta_{\rm app},
 \qquad
 D:=\delta_{\rm app}\left(\frac{R_s}{r_s}+2\right).
\]
Since $D\ge\delta_{\rm app}$, equations~\eqref{eq:direct-support-value} and~\eqref{eq:direct-approximate-accuracy} again show that $\abs{\widetilde Z-Z(Q,v)}\le1/200$.

In either setting, each invocation succeeds with probability at least $2/3$ for every fixed $Q$ and $v$. Lemma~\ref{lem:direct-classifier} therefore turns $125$ independent invocations with the same objective $\overline v$ into a determinant predictor with Haar-average success probability greater than $1/2$. The random sign vector $v$ and the internal randomness of the optimizer are independent of $Q$, so the purification in Theorem~\ref{thm:Fourier-rank-method} includes them, and the predictor is a query algorithm of the form \eqref{eq:standard-query-algorithm}. If the grid used by the optimizer depends on its objective, the predictor queries the union of the grids over the $2^d$ possible objectives $\overline v$; this union is finite and independent of $Q$. The predictor uses at most $125T$ membership queries to the same body on this grid. Theorem~\ref{thm:binary-determinant} therefore gives
\[
 125T\ge\frac d2,
\]
which is equivalent to $T\ge d/250$. Since $d=2\lfloor n/2\rfloor\ge n/2$ for $n\ge2$, we obtain $T\ge n/500$. For a continuous membership register, the same conclusion follows from Corollary~\ref{cor:continuous-membership-determinant}, because the maps $y\mapsto q_y$ and $y\mapsto\tau_y$ in \eqref{eq:membership-parameters} are continuous.
\end{proof}

The scale $s$ only rescales the ellipsoids, and the accuracy requirements in Theorem~\ref{thm:direct-membership-lower-bound} scale with it. The choice $s=\sqrt{a_d}$ makes the outer radius equal to one, which is the normalization used to compare with other lower bounds in the Discussion.

\begin{corollary}[Linear lower bounds for ellipsoids in the unit ball]
\label{cor:direct-normalized}
Let $n\ge2$ and consider the family $\{K_{\sqrt{a_d},Q}:Q\in\mathcal Q_n\}$. Every body $K$ in this family satisfies
\begin{equation}
 B_2(0,r)\subseteq K\subseteq B_2(0,1),
 \qquad
 r=\sqrt{\frac{a_d}{1+a_d}}=\Theta(n^{-1}).
 \label{eq:direct-normalized-geometry}
\end{equation}
An exactly feasible optimizer with accuracy
\begin{equation}
 \eps\le\frac1{200\sqrt n}
 \label{eq:direct-normalized-accuracy}
\end{equation}
that works for every body in this family and every unit objective requires at least $n/500$ exact membership queries. The same bound holds for an approximately feasible optimizer with tolerance
\begin{equation}
 \delta_{\rm app}\le\frac{r}{200\sqrt n\,(1+2r)}=\Theta(n^{-3/2}).
 \label{eq:direct-normalized-tolerance}
\end{equation}
\end{corollary}

\begin{proof}
With $s=\sqrt{a_d}$, the radii in Theorem~\ref{thm:direct-membership-lower-bound} are $R_s=1$ and $r_s=r$, and $r=\Theta(n^{-1})$ because $a_d=2g_d=\Theta(d^{-2})$ and $d=\Theta(n)$. The right side of \eqref{eq:direct-exact-accuracy} equals $1/(200\sqrt d)$, which is at least $1/(200\sqrt n)$ because $d\le n$. The right side of \eqref{eq:direct-approximate-accuracy} equals
\[
 \frac1{200\sqrt d\,(1/r+2)}=\frac{r}{200\sqrt d\,(1+2r)}\ge\frac{r}{200\sqrt n\,(1+2r)}.
\]
Theorem~\ref{thm:direct-membership-lower-bound} therefore applies in both cases.
\end{proof}

Theorem~\ref{thm:direct-membership-lower-bound} also covers the scaled families $s=n$ and $s=n^2$ used in the first version of this paper. Corollary~\ref{cor:direct-original-scales} in Section~\ref{sec:inverse-iteration-argument} records this comparison.

\section{First-order optimization of smooth and strongly convex functions}
\label{sec:strongly-convex-quadratics}

The same hard matrices give a lower bound in the first-order oracle model. The dimension of the hard block can be chosen as a function of the condition number and padded by known coordinates. This produces the usual transition at $\sqrt\kappa=n$.

For a differentiable function $f:\R^n\to\R$, define its ideal gradient phase oracle by
\begin{equation}
 (\mathcal G_f\psi)(x,u)
 :=e^{iu^T\nabla f(x)}\psi(x,u),
 \qquad
 (x,u)\in\R^n\times\R^n.
 \label{eq:gradient-phase-oracle}
\end{equation}
One application of $\mathcal G_f$ or $\mathcal G_f^{-1}$ counts as one gradient query. In the ideal continuous-register model, this oracle is equivalent, by the Fourier transform on the answer register, to the translation oracle
\[
 \ket{x,y}\longmapsto\ket{x,y+\nabla f(x)}.
\]
As in Remark~\ref{rem:ideal-continuous-translation}, this equivalence is not a claim about finite-precision encodings.

For a quadratic function with Hessian $H$, the inequalities
\[
 I\preceq H\preceq\kappa I
\]
mean that the function is $1$-strongly convex and $\kappa$-smooth. Its condition number is at most $\kappa$.

\subsection{The hard family of quadratics}
\label{sec:quadratic-family}

We first describe the hard functions for large $q:=\min\{\sqrt\kappa,n\}$; small $q$ is handled separately in each proof. Throughout this subsection, assume $q\ge200$. We use the construction of Section~\ref{sec:determinant-to-optimization} in a smaller even dimension. Set
\begin{equation}
 d:=2\left\lfloor
 \frac{\alpha}{2\pi}\min\{\sqrt\kappa,n\}
 \right\rfloor
 =2\left\lfloor\frac{q}{200}\right\rfloor.
 \label{eq:quadratic-hard-dimension}
\end{equation}
Since $q\ge200$, we have $2\le d<n$, $\kappa\ge q^2\ge2$, and
\begin{equation}
 \frac{q}{200}\le d\le\frac{q}{100},
 \qquad\text{so}\qquad
 d=\Theta(q).
 \label{eq:quadratic-hard-dimension-order}
\end{equation}
Let $Q\in\mathsf O(d)$, and form the matrix $A_Q\in\R^{d\times d}$ using \eqref{eq:hard-matrices} in dimension $d$. Thus
\begin{equation}
 2g_dI_d\preceq A_Q\preceq(1+2g_d)I_d.
 \label{eq:quadratic-hard-spectrum}
\end{equation}
Define the padded, normalized Hessian
\begin{equation}
 \widehat A_{n,Q}
 :=\frac{A_Q}{2g_d}\oplus\kappa I_{n-d}.
 \label{eq:quadratic-Hessian}
\end{equation}

The identity $1-\cos t=2\sin^2(t/2)$ and the bound $\sin x\ge2x/\pi$ for $x\in[0,\pi/2]$ give
\[
 g_d\ge\frac{\alpha^2}{\pi^2d^2}.
\]
Consequently,
\begin{align*}
 \lambda_{\max}\!\left(\frac{A_Q}{2g_d}\right)
 &\le\frac{1+2g_d}{2g_d}
 =1+\frac1{2g_d}\\
 &\le1+\frac{\pi^2d^2}{2\alpha^2}
 \le1+\frac\kappa2
 \le\kappa,
\end{align*}
where the penultimate inequality follows from $d\le(\alpha/\pi)\sqrt\kappa$, and the last holds for $\kappa\ge2$. Together with \eqref{eq:quadratic-hard-spectrum}, this proves
\begin{equation}
 I_n\preceq\widehat A_{n,Q}\preceq\kappa I_n.
 \label{eq:quadratic-conditioned-Hessian}
\end{equation}

For $b\in\mathbb S^{d-1}$, let $\overline b=(b,0_{n-d})$ and define
\begin{equation}
 f_{Q,b}(x)
 :=\frac12x^T\widehat A_{n,Q}x
    -\frac1{2g_d}\overline b^Tx,
 \qquad x\in\R^n.
 \label{eq:hard-quadratic-function}
\end{equation}
Equation~\eqref{eq:quadratic-conditioned-Hessian} shows that every such function is $1$-strongly convex and $\kappa$-smooth.

We next relate optimization to inverse application. The unique minimizer of \eqref{eq:hard-quadratic-function} is
\begin{equation}
 x_{Q,b}^*
 =\widehat A_{n,Q}^{-1}\frac{\overline b}{2g_d}
 =(A_Q^{-1}b,0_{n-d}).
 \label{eq:hard-quadratic-minimizer}
\end{equation}
For every $x\in\R^n$,
\begin{equation}
 f_{Q,b}(x)-f_{Q,b}(x_{Q,b}^*)
 =\frac12\norm{x-x_{Q,b}^*}_{\widehat A_{n,Q}}^2
 \ge\frac12\norm{x-x_{Q,b}^*}^2.
 \label{eq:quadratic-gap-identity}
\end{equation}
Consequently, for every $\eta>0$, a point $\widehat x$ with additive objective error at most $\eta^2/2$ satisfies
\begin{equation}
 \norm{\widehat x-x_{Q,b}^*}\le\eta.
 \label{eq:quadratic-inverse-accuracy}
\end{equation}
Projection onto the first $d$ coordinates therefore produces a vector within distance $\eta$ of $A_Q^{-1}b$.

It remains to count the hidden-matrix queries used by one gradient query. Write $x=(x',x'')$ and $u=(u',u'')$ according to the decomposition $\R^n=\R^d\oplus\R^{n-d}$. From \eqref{eq:A-explicit} and \eqref{eq:hard-quadratic-function},
\begin{align}
 u^T\nabla f_{Q,b}(x)
 ={}&\frac{\frac12+2g_d}{2g_d}(u')^Tx'
        +\kappa(u'')^Tx''
     -\frac1{2g_d}(u')^Tb \notag\\
 &+\frac1{8g_d}
 \left((u')^TQx'+(x')^TQu'\right).
 \label{eq:gradient-phase-decomposition}
\end{align}
The first three terms are independent of $Q$ and can be implemented without a query. The two remaining terms are matrix phases with rank-one frequencies. They can be implemented with two queries to $\mathcal P_Q$, using the query labels
\[
 \left(x',\frac{u'}{8g_d}\right)
 \qquad\text{and}\qquad
 \left(u',\frac{x'}{8g_d}\right).
\]
The same simulation with inverse phases implements $\mathcal G_{f_{Q,b}}^{-1}$. Hence every gradient query uses at most two matrix phase queries to $Q$.

\subsection{A linear gradient-query lower bound}
\label{sec:direct-gradient}

The random-objective predictor of Section~\ref{sec:direct-predictor} also applies to the quadratics $f_{Q,v}$ with $v$ a normalized sign vector. The minimizer $x^*_{Q,v}$ satisfies $v^T(x^*_{Q,v})_{[d]}=v^TA_Q^{-1}v$, where $(\cdot)_{[d]}$ denotes the first $d$ coordinates. An approximate minimizer therefore gives an estimate of $Z(Q,v)$.

\begin{theorem}[Linear gradient-query lower bound]
\label{thm:direct-gradient-lower-bound}
Let $n\ge1$ and $\kappa\ge1$, and put $q:=\min\{\sqrt\kappa,n\}$ and $\eps_{\rm quad}:=1/200$. There is an explicit family $\mathscr F_{n,\kappa}$ of quadratic functions on $\R^n$, each $1$-strongly convex and $\kappa$-smooth, such that any quantum algorithm which, for every $f\in\mathscr F_{n,\kappa}$, returns $\widehat x$ satisfying
\begin{equation}
 f(\widehat x)\le\min_{x\in\R^n}f(x)+\eps_{\rm quad}
 \label{eq:quadratic-constant-accuracy}
\end{equation}
with probability at least $2/3$ makes at least
\begin{equation}
 \frac{q}{10^5}
 \label{eq:quadratic-gradient-lower-bound}
\end{equation}
queries to the gradient phase oracle \eqref{eq:gradient-phase-oracle} in the worst case.
\end{theorem}

\begin{proof}
Put $\eta_0:=1/10$, so that $\eps_{\rm quad}=\eta_0^2/2$.

We first treat the range $q<200$. Let $e_1$ be the first standard basis vector and set $\mathscr F_{n,\kappa}:=\{f_+,f_-\}$, where
\[
 f_\pm(x):=\frac12\norm{x\mp2\eta_0e_1}^2.
\]
Both functions are $1$-strongly convex and $1$-smooth, hence $\kappa$-smooth. By the argument behind \eqref{eq:quadratic-inverse-accuracy}, the success condition requires the output to lie within distance $\eta_0$ of the corresponding minimizer $\pm2\eta_0e_1$. These two success regions are disjoint. A zero-query algorithm has the same output distribution for both functions and therefore cannot succeed with probability at least $2/3$ on each. At least one query is therefore necessary, and $1>q/10^5$ in this range.

Now let $q\ge200$, let $d$ be as in \eqref{eq:quadratic-hard-dimension}, and set
\[
 \mathscr F_{n,\kappa}
 :=\{f_{Q,b}:Q\in\mathsf O(d),\ b\in\mathbb S^{d-1}\}.
\]
Suppose that an algorithm satisfying the theorem uses at most $T$ gradient queries. Fix $Q\in\mathsf O(d)$ and a normalized sign vector $v=d^{-1/2}\sigma$, and run the algorithm on $f_{Q,v}$. Let $\widehat x_{[d]}$ denote the first $d$ coordinates of its output $\widehat x$, and define
\begin{equation}
 \widetilde Z
 :=\sqrt{a_d d\max\{v^T\widehat x_{[d]},0\}},
 \label{eq:direct-gradient-estimator}
\end{equation}
with $a_d=2g_d$ as in \eqref{eq:direct-scale}. On a successful invocation, \eqref{eq:hard-quadratic-minimizer} and \eqref{eq:quadratic-inverse-accuracy} with $\eta=\eta_0$ give $\norm{\widehat x_{[d]}-A_Q^{-1}v}\le\eta_0$. For real numbers $\beta$ and $\beta'\ge0$, we have $\abs{\max\{\beta,0\}-\beta'}\le\abs{\beta-\beta'}$. Applying this with $\beta=v^T\widehat x_{[d]}$ and $\beta'=v^TA_Q^{-1}v$, and then using Cauchy--Schwarz with $\norm v=1$, we obtain
\[
 \abs{\max\{v^T\widehat x_{[d]},0\}-v^TA_Q^{-1}v}
 \le\norm{\widehat x_{[d]}-A_Q^{-1}v}
 \le\eta_0.
\]
For nonnegative $x$ and $y$, $\abs{\sqrt x-\sqrt y}\le\sqrt{\abs{x-y}}$. It follows that
\begin{equation}
 \abs{\widetilde Z-Z(Q,v)}
 \le\sqrt{a_d d\,\eta_0}
 \le\sqrt{\frac{\alpha^2\eta_0}{2d}}
 \le\sqrt{\frac{\alpha^2}{40}}
 <\frac1{200}.
 \label{eq:direct-gradient-error}
\end{equation}
Here the second inequality uses $a_d\le\alpha^2/(2d^2)$ from \eqref{eq:direct-A}, the third uses $d\ge2$, and the last holds because $\alpha^2/40=\pi^2/400000<1/40000$.

Each invocation therefore satisfies the hypothesis of Lemma~\ref{lem:direct-classifier} for every fixed $Q$ and $v$. That lemma turns $125$ independent invocations on the same function $f_{Q,v}$ into a determinant predictor with Haar-average success probability greater than $1/2$. Each invocation uses at most $T$ gradient phase queries, and each gradient phase query uses at most two matrix phase queries to $Q$, as shown after \eqref{eq:gradient-phase-decomposition}. The phases involving the known vector $v$ require no oracle calls. The predictor therefore uses at most $250T$ matrix phase queries, and Theorem~\ref{thm:mv-determinant} in dimension $d$ gives
\[
 250T\ge\frac d2.
\]
Hence $T\ge d/500\ge q/10^5$ by \eqref{eq:quadratic-hard-dimension-order}.
\end{proof}

\paragraph{Comparison with accelerated methods.}
Accelerated gradient descent uses
\[
 O\!\left(\sqrt\kappa\,
 \log\frac{f(x_0)-f_*}{\eps}\right)
\]
gradient evaluations for smooth strongly convex optimization~\cite{nesterov_acceleration}. For the family above, taking $x_0=0$ gives an initial objective gap polynomial in $\kappa$, while the required accuracy is constant. Accelerated gradient descent therefore uses $O(\sqrt\kappa\log\kappa)$ queries. For quadratics, conjugate-gradient methods also give an $O(n)$ exact-arithmetic upper bound. Theorem~\ref{thm:direct-gradient-lower-bound} matches the resulting $\min\{\sqrt\kappa,n\}$ dependence up to the factor $\log\kappa$ in the accelerated upper bound, in the ideal quantum gradient-query model.

The normalization in \eqref{eq:quadratic-Hessian} converts the original $\Theta(n^{-2})$ accuracy scale into constant absolute accuracy. It also rescales the linear term and the initial objective gap. Thus the theorem does not identify precision and condition number in every normalization; it gives two equivalent descriptions of this hard family.

\section{The inverse-iteration argument from the first version}
\label{sec:inverse-iteration-argument}

This section presents the reduction from the first version of this paper. It proves the optimization lower bounds and the gradient-query lower bound up to a factor of order $\log n\log\log n$. We retain it because it is the argument of the first version, and the rest of this section follows that version with only minor changes. The reduction uses an accurate optimizer to apply an approximate inverse of $A_Q$, and then uses inverse iteration to find the eigenvector of $A_Q$ with the smallest eigenvalue. In contrast with Section~\ref{sec:linear-lower-bound}, it needs $O(\log d\,\log\log d)$ optimizer invocations with adaptively chosen objectives, and it uses the events $\mathcal E_{d,\pm}$ from Lemma~\ref{lem:regular-events}. Sections~\ref{sec:optimizer-inverse} and~\ref{sec:inverse-iteration} construct the determinant predictor. Sections~\ref{sec:exactly-feasible-lower-bound} and~\ref{sec:approximately-feasible-optimization} choose the scales and prove the membership-query lower bounds. Section~\ref{sec:quadratic-inverse-iteration} proves the gradient-query lower bound for strongly convex quadratics.

\subsection{An optimizer approximately applies the inverse}
\label{sec:optimizer-inverse}

We next make reductions from optimization to the required linear algebraic primitives. For a positive-definite matrix $B$, write $K_B=\{x:x^TBx\le1\}$ and $\norm{x}_B=\sqrt{x^TBx}$. The unique maximizer of $b^Tx$ over $K_B$ is a normalized copy of $B^{-1}b$. If an optimizer returns a nearly optimal feasible point $x$, its objective value estimates the missing normalization. The vector $(b^Tx)x$ therefore approximates $B^{-1}b$. The following lemma makes this statement quantitative; the formula for the maximum follows from standard convex duality~\cite{boyd2004convex}.

\begin{lemma}[A near-maximizer applies the inverse]
\label{lem:maximizer-inverse}
Let $\underline\lambda>0$, let $B$ be a real symmetric matrix satisfying $B\succeq\underline\lambda I$, let $b$ be a unit vector, and put
\[
 h:=\sqrt{b^TB^{-1}b},
 \qquad
 x_*:=\frac{B^{-1}b}{h}.
\]
Then $h_{K_B}(b)=h$.  If $x\in K_B$ and $b^Tx\ge h-\eps$, then the computable vector $z=(b^Tx)x$ satisfies
\begin{equation}
 \norm{z-B^{-1}b}_B^2
 \le\frac{4\eps}{\sqrt{\underline\lambda}}+2\eps^2.
 \label{eq:inverse-error}
\end{equation}
\end{lemma}

\begin{proof}
Cauchy--Schwarz applied to $B^{1/2}x$ and $B^{-1/2}b$ gives the formula for the maximum, attained at $x_*$. Put $u=B^{1/2}x$ and $a=B^{-1/2}b$. Then $\norm u\le1$, $\norm a=h$, and
\[
 \norm{x-x_*}_B^2
 =\norm{u-a/h}^2
 \le2-2\frac{a^Tu}{h}
 \le\frac{2\eps}{h}.
\]
As $0\le h-b^Tx\le\eps$ and $B^{-1}b=h x_*$,
\[
 \norm{z-B^{-1}b}_B
 \le h\norm{x-x_*}_B+\eps\norm x_B
 \le\sqrt{2\eps h}+\eps.
\]
Use $h\le\underline\lambda^{-1/2}$ and $(r+s)^2\le2r^2+2s^2$.
\end{proof}

We will apply the lemma with $B=A_Q$. The next constants specify the required optimization accuracy and the resulting inverse errors. Here $\eps_{0,d}$ is the objective error on the unscaled ellipsoid, $\rho$ is the error from one successful optimizer call, and $E_0=3\rho$ is the error after the amplification used below. Define
\begin{equation}
 c_\star:=10^{-8},
 \qquad
 \eps_{0,d}:=c_\star g_d^{3/2},
 \qquad
 \rho:=\sqrt{\sqrt2c_\star+c_\star^2},
 \qquad
 E_0:=3\rho.
 \label{eq:inverse-constants}
\end{equation}
The numerical choice $c_\star=10^{-8}$ gives $E_0<10^{-3}$.

\begin{corollary}[One optimizer call gives a constant-error inverse]
\label{cor:one-inverse}
Suppose an exactly feasible optimizer for $K_{s,Q}$ works for every unit objective and has objective error at most $s\eps_{0,d}$. For every unit $b\in\R^d$, one invocation and classical postprocessing return a vector $z_b$ such that
\begin{equation}
 \Prb\!\left[\norm{z_b-A_Q^{-1}b}\le\rho\right]\ge\frac23.
 \label{eq:one-inverse}
\end{equation}
The postprocessing uses no additional oracle calls.
\end{corollary}

\begin{proof}
Invoke the optimizer with objective $(b,0_{\ell_n})$ and divide the first $d$ coordinates of its output by $s$, obtaining $x$. On the optimizer's success event, the definition of $K_{s,Q}$ gives $x\in K_{A_Q}$, and the objective guarantee gives $b^Tx\ge h_{K_{A_Q}}(b)-\eps_{0,d}$. Apply Lemma~\ref{lem:maximizer-inverse} with $\underline\lambda=2g_d$ and set $z_b=(b^Tx)x$. It gives
\[
 \norm{z_b-A_Q^{-1}b}_{A_Q}^2
 \le2\sqrt2c_\star g_d+2c_\star^2g_d^3.
\]
Dividing by the spectral lower bound $2g_d$ gives
\[
 \norm{z_b-A_Q^{-1}b}^2
 \le\sqrt2c_\star+c_\star^2g_d^2
 \le\sqrt2c_\star+c_\star^2=\rho^2.
\]
\end{proof}

\subsection{Using inverse iteration to predict the determinant}
\label{sec:inverse-iteration}

Corollary~\ref{cor:one-inverse} turns one optimizer invocation into an approximate application of $A_Q^{-1}$, but only with probability $2/3$. We will boost the success probability by repeating each invocation and select a cluster of mutually close outputs to make every inverse step reliable. On $\mathcal E_{d,-}$, a random sign vector has overlap at least $1/\sqrt{2d}$ with the exceptional eigenvector $u$ with constant probability. Each application of $A_Q^{-1}$ amplifies this component relative to the orthogonal component, and $O(\log d)$ iterations make the resulting vector nearly parallel to $u$. A final estimate of the maximum objective value then distinguishes this case from $\mathcal E_{d,+}$.

We first amplify a single approximate inverse application. Note that the selection rule below does not need to know the target vector $z_*$: it chooses an output surrounded by a strict majority of the other outputs.

\begin{lemma}[Majority-ball amplification]
\label{lem:majority-ball}
Let $m$ be odd.  Suppose independent outputs $z_1,\ldots,z_m$ each lie within radius $r$ of a fixed $z_*$ with probability at least $2/3$.  Select the least-indexed output whose closed $2r$-ball contains more than $m/2$ outputs, if such an output exists.  Whenever a strict majority is successful, the selected output lies within $3r$ of $z_*$.  The probability of no strict successful majority is at most $e^{-m/18}$.
\end{lemma}

\begin{proof}
All successful outputs are pairwise within $2r$, so every successful output is an eligible center when they form a strict majority. The selected ball then intersects the successful set, and the triangle inequality gives the $3r$ bound. Let $X_i:=\1[\norm{z_i-z_*}\le r]$. Hoeffding's inequality~\cite{hoeffding1963probability} gives
\[
 \Prb\!\left[\sum_{i=1}^mX_i\le m/2\right]
 \le e^{-2m(2/3-1/2)^2}=e^{-m/18}.
\]
\end{proof}

We next choose the initial vector. A random sign vector has enough overlap with any fixed direction with constant probability.

\begin{lemma}[Rademacher initialization]
\label{lem:rademacher}
For every unit $u\in\R^d$, let $v_0=d^{-1/2}(\sigma_1,\ldots,\sigma_d)$ with independent uniform signs. Then
\begin{equation}
 \Prb\!\left[\abs{u^Tv_0}^2\ge\frac1{2d}\right]\ge\frac1{12}.
 \label{eq:Rademacher-overlap}
\end{equation}
\end{lemma}

\begin{proof}
For $X=\abs{u^Tv_0}^2$, Lemma~\ref{lem:rademacher-moments} applied to $dX$ gives
\[
 \E X=\frac1d,
 \qquad
 \E X^2=\frac1{d^2}
 \left(3\Big(\sum_i u_i^2\Big)^2-2\sum_i u_i^4\right)
 \le\frac3{d^2}.
\]
Paley--Zygmund with parameter $1/2$ gives the result. This is the standard second-moment lower-tail inequality of Paley and Zygmund~\cite{paley1932some}.
\end{proof}

From \eqref{eq:inverse-gap}, on $\mathcal E_{d,-}$, the eigenvalue of $A_Q$ in direction $u$ is $2g_d$, while every eigenvalue on $u^\perp$ is at least $3g_d$. The next lemma shows that inverse iteration can still exploit this difference when each approximate inverse application has additive error.

\begin{lemma}[Noisy inverse iteration]
\label{lem:noisy-inverse}
Let $Q\in\mathcal E_{d,-}$ and let $u$ be the unit vector in \eqref{eq:inverse-gap}.  Define
\begin{equation}
 k:=\left\lceil\log_{4/3}(3\sqrt{2d})\right\rceil.
 \label{eq:iteration-count}
\end{equation}
Starting from the Rademacher vector $v_0$, suppose
\[
 y_j=A_Q^{-1}v_j+e_j,
 \qquad
 \norm{e_j}\le E_0,
 \qquad
 v_{j+1}=\frac{y_j}{\norm{y_j}},
 \quad 0\le j<k.
\]
On the event in \eqref{eq:Rademacher-overlap}, every normalization is defined and
\begin{equation}
 \abs{u^Tv_k}^2\ge\frac9{10}.
 \label{eq:final-overlap}
\end{equation}
\end{lemma}

\begin{proof}
Inverse iteration and its shift-and-invert variants are classical methods for isolating extremal eigenspaces~\cite{saad2011numerical, garber2016faster}. Here, we determine the additive-error estimate needed for this oracle reduction. We track the overlap $a_j$ with the target direction, the ratio $r_j$ of the orthogonal component to the target component, and the relative inverse error $\zeta_j$. Thus let
\[
 a_j:=\abs{u^Tv_j},
 \qquad
 r_j:=\frac{\norm{(I-uu^T)v_j}}{a_j},
 \qquad
 \zeta_j:=\frac{2g_dE_0}{a_j}.
\]
By \eqref{eq:inverse-gap}, the target component of $A_Q^{-1}v_j$ has magnitude $a_j/(2g_d)$ and the orthogonal component has norm at most $a_jr_j/(3g_d)$.  Whenever $\zeta_j<1$,
\begin{equation}
 r_{j+1}\le\frac{(2/3)r_j+\zeta_j}{1-\zeta_j}.
 \label{eq:ratio-recurrence}
\end{equation}
The same target-component lower bound shows $y_j\ne0$.

We prove inductively that $a_j\ge1/\sqrt{2d}$ and $\zeta_j<1/100$.  Since $g_d\le\alpha^2/(4d^2)$, the first inequality implies
\[
 \zeta_j\le2\frac{\alpha^2}{4d^2}E_0\sqrt{2d}
 \le\frac{\alpha^2E_0}{4}<\frac1{100}.
\]
If $r_j\ge1/3$, recurrence \eqref{eq:ratio-recurrence} and $\zeta_j<1/45$ imply $r_{j+1}\le(3/4)r_j$.  If $r_j\le1/3$, the same recurrence and $\zeta_j<1/12$ imply $r_{j+1}\le1/3$.  Since $a_j^2=1/(1+r_j^2)$, these alternatives preserve $a_j\ge1/\sqrt{2d}$.

Initially $r_0<\sqrt{2d}$.  From the definition of $k$, after repeated contractions by 3/4, we reach $r_k\le1/3$, and hence $\abs{u^Tv_k}^2=1/(1+r_k^2)\ge9/10$.
\end{proof}

We now assemble the preceding lemmas. There are $k$ amplified inverse steps followed by one amplified estimate of the maximum objective value, so the wrapper uses $k+1$ batches of optimizer invocations.

\begin{proposition}[Predicting determinant using optimizer invocations]
\label{prop:optimizer-determinant}
Suppose the optimizer in Corollary~\ref{cor:one-inverse} uses at most $T$ membership queries per invocation.  Put
\begin{equation}
 m:=2\left\lceil9\log(100(k+1))\right\rceil+1,
 \qquad
 W:=(k+1)m.
 \label{eq:wrapper-parameters}
\end{equation}
Using $W$ fresh optimizer invocations, hence at most $M=WT$ membership queries, one can predict $\det Q$ for $Q\sim\mu_d$ with average success at least $1/2+\beta_\star$, where $\beta_\star:=1/50$. The only oracle calls are the membership queries made by the optimizer.
\end{proposition}

\begin{proof}
The reduction calls the same optimizer on several objectives. Each objective after the first depends on the outputs of earlier calls. This is why the optimizer must work for every unit objective. Once we condition on the earlier outputs, the next objective is fixed, so the optimizer still succeeds with probability at least $2/3$.

Choose $v_0$ as in Lemma~\ref{lem:rademacher}. At each of the $k$ inverse iterations, run the construction of Corollary~\ref{cor:one-inverse} independently $m$ times and apply the majority-ball selector with radius $\rho$. On a strict successful majority, the selected vector has additive error at most $E_0=3\rho$ and is nonzero because $\norm{A_Q^{-1}v_j}\ge(1+2g_d)^{-1}>1/3>E_0$. Normalize it to obtain the next iterate; use a fixed unit vector as a fallback if selection fails.

After $k$ iterations, run the optimizer independently another $m$ times with objective $(v_k,0_{\ell_n})$. Divide each output's first $d$ coordinates by $s$, call the result $x_i$, and take the median $\widetilde h$ of $v_k^Tx_i$. For every successful call, this value lies between $h_{K_{A_Q}}(v_k)-\eps_{0,d}$ and $h_{K_{A_Q}}(v_k)$. A strict successful majority therefore places the median in the same interval. Set
\begin{equation}
 \tau_d:=\sqrt{\frac{2}{5g_d}},
 \label{eq:objective-threshold}
\end{equation}
and output $-1$ exactly when $\widetilde h>\tau_d$.

Each batch fails to contain a strict successful majority with probability at most
\[
 e^{-m/18}\le\frac1{100(k+1)}.
\]
The same conditional guarantee holds at every stage. A union bound over the $k+1$ batches gives total failure probability at most
\[
 (k+1)\frac1{100(k+1)}=\frac1{100}.
\]
Thus all batches are good with probability at least $99/100$.

If $Q\in\mathcal E_{d,+}$, then $A_Q\succeq3g_dI$ and every unit $v$ obeys
\[
 h_{K_{A_Q}}(v)^2=v^TA_Q^{-1}v
 \le\frac1{3g_d}<\tau_d^2.
\]
A strict successful majority in the final batch therefore makes the classifier output $+1$.

If $Q\in\mathcal E_{d,-}$, the Rademacher event and good inverse batches give $\abs{u^Tv_k}^2\ge9/10$.  Hence
\[
 h_{K_{A_Q}}(v_k)^2\ge\frac9{20g_d}
\]
and
\begin{equation}
 h_{K_{A_Q}}(v_k)-\tau_d
 \ge\frac{3-2\sqrt2}{\sqrt{20g_d}}
 >c_\star g_d^{3/2}=\eps_{0,d}.
 \label{eq:objective-margin}
\end{equation}
Thus a strict successful majority in the final batch makes the classifier output $-1$. The initialization event has probability at least $1/12$, while the probability that some batch fails is at most $1/100$. The negative regular component is therefore classified correctly with probability at least $1/12-1/100=11/150$.

On the positive component, regularity and good batches give success probability at least $(99/100)^2$. On the negative component, regularity together with the preceding estimate gives success probability at least $(99/100)(11/150)$. Since the two determinant components have equal Haar measure, the average success is at least
\[
 \frac12\frac{99}{100}
 \left(\frac{99}{100}+\frac{11}{150}\right)
 >\frac{13}{25}=\frac12+\frac1{50}.
\]
\end{proof}

Finally, $k=O(\log d)$ and $m=O(\log\log d)$, so
\[
 W=(k+1)m=O(\log d\,\log\log d).
\]
The classifier has success probability strictly greater than $1/2$, so Theorem~\ref{thm:binary-determinant} implies $WT\ge d/2$. Therefore
\[
 T=\Omega\!\left(\frac{d}{\log d\,\log\log d}\right).
\]
More explicitly, for $n\ge100$ the definitions give
\[
 k+1\le8\log n,
 \qquad
 m\le147\log\log n,
 \qquad
 W\le1176\log n\log\log n.
\]
Since $d\ge n/2$, the determinant lower bound therefore implies
\begin{equation}
 T\ge\frac{n}{4704\log n\log\log n}.
 \label{eq:ambient-wrapper-bound}
\end{equation}
The next two parts of this subsection choose the scale $s$ and translate this bound into the stated optimization accuracy and geometric parameters.

Proposition~\ref{prop:optimizer-determinant} proved the query lower bound once the optimizer is accurate enough to simulate inverse iteration. It remains to choose the scale of the ellipsoids and verify the final accuracy and geometric parameters. The unscaled objective error required for exactly feasible optimization is $\Theta(n^{-3})$, so scaling by $s=n$ gives error $\Theta(n^{-2})$. Approximately feasible optimization requires the smaller unscaled tolerance $\Theta(n^{-4})$, so we instead scale by $s=n^2$. This again gives a final tolerance of order $n^{-2}$. 

\subsection{Lower bound for exactly feasible optimization}
\label{sec:exactly-feasible-lower-bound}

We first take $s=n$. The theorem records the resulting accuracy, ball inclusions, and membership-query lower bound.

\begin{theorem}[Exactly feasible optimization lower bound]
\label{thm:exactly-feasible-main}
There are universal constants
\[
 n_0,c_\eps,C_\eps,C_{\rm lb}>0
\]
such that for every $n\ge n_0$ there is an explicit family $\mathscr K_n$ of $n$-dimensional ellipsoids and an accuracy $\eps_n$ satisfying
\[
 c_\eps n^{-2}\le\eps_n\le C_\eps n^{-2}.
\]
There are known parameters $r_n,R_n>0$ such that every $K\in\mathscr K_n$ satisfies
\begin{equation}
 B_2(0,r_n)\subseteq K\subseteq B_2(0,R_n),
 \quad
 R_n-r_n\ge\frac12R_n,
 \quad
 \frac{r_n\eps_n}{R_n}=O(n^{-3}).
 \label{eq:exactly-feasible-geometry}
\end{equation}
Let $\mathcal G\subseteq\R^n$ be any nonempty finite grid fixed in advance. Every exactly feasible optimizer of accuracy $\eps_n$ that works for all $K\in\mathscr K_n$ and all unit objectives, using exact membership on $\mathcal G$, requires
\begin{equation}
 T\ge C_{\rm lb}\frac{n}{\log n\,\log\log n}
 \label{eq:exactly-feasible-lower-bound}
\end{equation}
queries in the worst case. One may take $n_0=100$ and $C_{\rm lb}=1/4704$.
\end{theorem}

\begin{proof}
Set $s_n=n$ and define
\begin{equation}
 \eps_n:=nc_\star\overline g_n^{3/2},
 \qquad
 \mathscr K_n:=\{K_{n,Q}:Q\in\mathcal Q_n\}.
 \label{eq:exactly-feasible-family}
\end{equation}
We first verify the accuracy. The identity $1-\cos t=2\sin^2(t/2)$ and the standard linear bounds for sine give
\begin{equation}
 \frac{\alpha^2}{\pi^2n^2}\le \overline g_n\le\frac{\alpha^2}{n^2}
 \label{eq:g-bounds}
\end{equation}
These estimates imply $\eps_n=\Theta(n^{-2})$.

We next verify the geometry. Equation~\eqref{eq:scaled-ball-bounds} gives
\[
 r_n=\frac{n}{\sqrt{1+2\overline g_n}}=\Theta(n),
 \qquad
 R_n=\frac{n}{\sqrt{2\overline g_n}}=\Theta(n^2).
\]
Moreover,
\[
 \frac{r_n}{R_n}
 =\sqrt{\frac{2\overline g_n}{1+2\overline g_n}}\le\frac12
\]
for $n\ge100$. Hence $R_n-r_n\ge R_n/2$. The normalized accuracy satisfies
\[
 \frac{r_n\eps_n}{R_n}
 \le\sqrt2c_\star n\overline g_n^2=O(n^{-3}).
\]

It remains to prove the query lower bound. Suppose one optimizer invocation uses $T$ queries. Proposition~\ref{prop:optimizer-determinant} converts it into a determinant classifier, and the explicit ambient-dimensional estimate \eqref{eq:ambient-wrapper-bound} gives
\[
 T\ge\frac{n}{4704\log n\log\log n}.
\]
\end{proof}

\subsection{Lower bound for approximately feasible optimization}
\label{sec:approximately-feasible-optimization}

An approximately feasible output may lie outside the ellipsoid, so the near-maximizer argument from Section~\ref{sec:optimizer-inverse} no longer applies directly. We handle this in three steps. First, Lemma~\ref{lem:approximately-feasible-geometry} in Section~\ref{sec:optimization-guarantees} replaces the output analytically by a nearby feasible witness and bounds the resulting objective loss. Second, we use the original output to construct an approximate inverse application. Finally, we choose the tolerance small enough to preserve both inverse iteration and the final test of the maximum objective value. We use the convention from Definition~\ref{def:approximately-feasible-opt}, standard in geometric oracle theory~\cite{grotschel2012geometric,van2020convex}.

Choose the promised $x\in K$ with $\norm{x-y}\le\delta$. Then $c^Tx\ge c^Ty-\delta$, and $h_K(c)\le R$ gives \eqref{eq:approximately-feasible-effective-error}. The same lower estimate gives the left side of \eqref{eq:approximately-feasible-objective-interval}. The upper side follows from $c^Ty\le c^Tx+\delta\le h_K(c)+\delta$.
\end{proof}

We now adapt the near-maximizer calculation to the possibly infeasible output $y$. The scalar $q=b^Ty$ estimates $\max_{x\in K_B}b^Tx$, so $z=qy$ remains computable from the optimizer's output.

\begin{lemma}[An approximately feasible output gives a computable inverse]
\label{lem:approximately-feasible-inverse}
Let $0<\mu\le L$, let $B$ be a real symmetric matrix satisfying $\mu I\preceq B\preceq LI$, let $b$ be unit, and let $y$ be a successful approximately feasible output with tolerance $0<\delta<L^{-1/2}$ for $K_B$. Put
\[
 h=\sqrt{b^TB^{-1}b},
 \quad
 D=\delta\left(\sqrt{L/\mu}+2\right),
 \quad
 q=b^Ty,
 \quad
 z=qy.
\]
Then
\begin{equation}
 \norm{z-B^{-1}b}_B
 \le(h+D)\left(\sqrt L\,\delta+\sqrt{\frac{2D}{h}}\right)+D.
 \label{eq:approximately-feasible-inverse-error}
\end{equation}
\end{lemma}

\begin{proof}
The spectral bounds imply
\[
 B_2(0,L^{-1/2})\subseteq K_B\subseteq B_2(0,\mu^{-1/2}).
\]
Apply Lemma~\ref{lem:approximately-feasible-geometry} with $r=L^{-1/2}$ and $R=\mu^{-1/2}$, and let $x$ be the resulting feasible witness. If $x_*=B^{-1}b/h$, the proof of Lemma~\ref{lem:maximizer-inverse}, with $D$ in place of $\eps$, gives
\[
 \norm{x-x_*}_B\le\sqrt{\frac{2D}{h}},
 \qquad
 \norm{y-x}_B\le\sqrt L\delta.
\]
Equation~\eqref{eq:approximately-feasible-objective-interval} gives $\abs{q-h}\le D$ and $\abs q\le h+D$. Since $B^{-1}b=hx_*$ and $\norm{x_*}_B=1$, the decomposition
\[
 qy-hx_*=q(y-x)+q(x-x_*)+(q-h)x_*
\]
and the triangle inequality give \eqref{eq:approximately-feasible-inverse-error}.
\end{proof}

The final threshold for the maximum objective value lies strictly between the two regular cases. The positive constants
\begin{equation}
 \Delta_+:=\sqrt{\frac25}-\frac1{\sqrt3},
 \qquad
 \Delta_-:=\frac{3-2\sqrt2}{\sqrt{20}}
 \label{eq:approximately-feasible-margins}
\end{equation}
are the corresponding gaps after factoring out $1/\sqrt{\overline g_n}$. We choose one constant small enough to preserve both these gaps and the inverse-error bound:
\begin{equation}
 \gamma:=\min\left\{1,\frac{\rho^2}{49},
 \Delta_+,\frac{\Delta_-}{3}\right\},
 \label{eq:approximately-feasible-constant}
\end{equation}
and set
\begin{equation}
 \delta_{0,n}:=\gamma\overline g_n^2.
 \label{eq:approximately-feasible-unscaled-accuracy}
\end{equation}

\begin{lemma}[Approximate feasibility preserves the reduction]
\label{lem:approximately-feasible-margins}
Let $Q\in\mathcal Q_n$, put $d=d_n$ and $B=\overline A_{n,Q}$, and let $\pi_d:\R^n\to\R^d$ denote projection onto the first $d$ coordinates. For a unit vector $b\in\R^d$, set $\overline b=(b,0_{\ell_n})$. Suppose $y$ is a successful approximately feasible output for $K_B$ with objective $\overline b$ and tolerance $\delta_{0,n}$. Put
\[
 q=\overline b^Ty,
 \qquad
 z=qy.
\]
Then
\[
 \norm{\pi_dz-A_Q^{-1}b}\le\rho.
\]
Moreover, using objectives of the form $\overline b$, the final median test at threshold \eqref{eq:objective-threshold} correctly separates the regular determinant components on a strict successful majority.
\end{lemma}

\begin{proof}
Apply Lemma~\ref{lem:approximately-feasible-inverse} in dimension $n$ with objective $\overline b$. Here $\mu=2\overline g_n$, $L=1+2\overline g_n\le2$, and
\[
 D\le3\gamma\overline g_n^{3/2},
 \qquad
 L^{-1/2}\le h\le(2\overline g_n)^{-1/2}.
\]
Substitution into \eqref{eq:approximately-feasible-inverse-error} gives
\begin{align*}
 \norm{z-B^{-1}\overline b}_B
 &\le5\gamma\overline g_n^{3/2}+2\sqrt{3\sqrt2\gamma\overline g_n},\\
 \norm{z-B^{-1}\overline b}
 &\le\frac5{\sqrt2}\gamma\overline g_n+\sqrt{6\sqrt2\gamma}
 \le7\sqrt\gamma\le\rho.
\end{align*}
Since
\[
 B^{-1}\overline b
 =\overline A_{n,Q}^{-1}(b,0_{\ell_n})
 =(A_Q^{-1}b,0_{\ell_n}),
\]
and Euclidean projection does not increase distance, we obtain
\[
 \norm{\pi_dz-A_Q^{-1}b}
 \le\norm{z-B^{-1}\overline b}
 \le\rho.
\]
Thus the projected vectors $\pi_dz$ have exactly the inverse-error guarantee required by the amplification and inverse iteration in Section~\ref{sec:inverse-iteration}.

For the final maximum-value test, padding the objective does not change the exact maximum:
\[
 \sqrt{\overline b^TB^{-1}\overline b}
 =\sqrt{b^TA_Q^{-1}b}.
\]
On $\overline{\mathcal E}_{n,+}$, the gap between the upper bound on this maximum and the threshold in \eqref{eq:objective-threshold} is $\Delta_+/\sqrt{\overline g_n}$; successful approximately feasible outputs have objective value at most $h+\delta_{0,n}$. On $\overline{\mathcal E}_{n,-}$ after inverse iteration, the gap between the lower bound on the maximum and the threshold is $\Delta_-/\sqrt{\overline g_n}$; successful approximately feasible outputs have objective value at least $h-D$. The definition of $\gamma$ and $\overline g_n<1$ give
\[
 \delta_{0,n}<\frac{\Delta_+}{\sqrt{\overline g_n}},
 \qquad
 D<\frac{\Delta_-}{\sqrt{\overline g_n}},
\]
so the same median threshold separates the cases.
\end{proof}

The unscaled tolerance is $\delta_{0,n}=\Theta(n^{-4})$. We therefore scale the ellipsoid by $n^2$: distances and erosions scale linearly, so this produces the target tolerance $\Theta(n^{-2})$ without changing the membership-query reduction.

\begin{theorem}[Approximately feasible optimization lower bound]
\label{thm:approximately-feasible-main}
For every sufficiently large $n$, set
\[
 \delta_n:=n^2\delta_{0,n}=\Theta(n^{-2}),
\]
and consider the family
\begin{equation}
 \mathscr K_n^{\rm app}:=\{K_{n^2,Q}:Q\in\mathcal Q_n\}.
 \label{eq:approximately-feasible-family}
\end{equation}
For every $K\in\mathscr K_n^{\rm app}$,
\[
 B_2(0,r_n^{\rm app})\subseteq K\subseteq B_2(0,R_n^{\rm app}),
 \qquad
 r_n^{\rm app}=\Theta(n^2),
 \qquad
 R_n^{\rm app}=\Theta(n^3).
\]
Moreover,
\[
 R_n^{\rm app}-r_n^{\rm app}\ge\frac12R_n^{\rm app},
 \qquad
 \frac{r_n^{\rm app}\delta_n}{R_n^{\rm app}}=O(n^{-3}).
\]
Every approximately feasible optimizer with tolerance $\delta_n$ that works for every $K\in\mathscr K_n^{\rm app}$ and every unit objective requires
\[
 \Omega\!\left(\frac{n}{\log n\,\log\log n}\right)
\]
exact membership queries in the worst case, on every nonempty finite grid fixed in advance.
\end{theorem}

\begin{proof}
Let $s=n^2$. Scaling commutes with erosion:
\[
 (sK)^{-s\delta}=s(K^{-\delta}).
\]
Thus an approximately feasible output with tolerance $s\delta_{0,n}$ for $K_{s,Q}=sK_{\overline A_{n,Q}}$ becomes an approximately feasible output with tolerance $\delta_{0,n}$ for $K_{\overline A_{n,Q}}$ after division by $s$. At every stage of the wrapper, use the padded objective $\overline b=(b,0_{\ell_n})$ and project the postprocessed proposal onto its first $d_n$ coordinates. Lemma~\ref{lem:approximately-feasible-margins} then gives the $d_n$-dimensional inverse accuracy and maximum-value separation required by Proposition~\ref{prop:optimizer-determinant}. The resulting determinant classifier has the same bias $\beta_\star$ and uses the same wrapper as in Section~\ref{sec:inverse-iteration}. The ambient-dimensional estimate \eqref{eq:ambient-wrapper-bound} proves the query lower bound.

It remains to verify the stated tolerance and geometry. Since $s=n^2$ and $\overline g_n=\Theta(n^{-2})$,
\[
 \delta_n=n^2\gamma\overline g_n^2=\Theta(n^{-2}),
 \quad
 r_n^{\rm app}=\frac{n^2}{\sqrt{1+2\overline g_n}}=\Theta(n^2),
 \quad
 R_n^{\rm app}=\frac{n^2}{\sqrt{2\overline g_n}}=\Theta(n^3).
\]
Moreover,
\[
 \frac{r_n^{\rm app}}{R_n^{\rm app}}
 =\sqrt{\frac{2\overline g_n}{1+2\overline g_n}}\le\frac12
\]
for all sufficiently large $n$, which gives $R_n^{\rm app}-r_n^{\rm app}\ge R_n^{\rm app}/2$. Finally,
\[
 \frac{r_n^{\rm app}\delta_n}{R_n^{\rm app}}
 =\delta_n\sqrt{\frac{2\overline g_n}{1+2\overline g_n}}
 =O(n^{-3}).
\]
\end{proof}

\begin{corollary}[Optimization with continuous exact membership]
\label{cor:continuous-membership-optimization}
The lower bounds in Theorems~\ref{thm:exactly-feasible-main} and~\ref{thm:approximately-feasible-main} remain valid when the finite membership register is replaced by exact coherent membership on $L^2(\R^n)\otimes\C^2$.
\end{corollary}

\begin{proof}
For the hard ellipsoids, the membership predicate has the form \eqref{eq:membership-threshold}, and the maps from the continuous query point $y$ to $q_y$ and $\tau_y$ are continuous. Corollary~\ref{cor:continuous-membership-determinant} therefore replaces Theorem~\ref{thm:binary-determinant} in the optimizer wrapper. Every other step of the reduction is unchanged.
\end{proof}

\subsection{Comparison with the linear lower bound}

Theorem~\ref{thm:direct-membership-lower-bound} removes the factor $\log n\log\log n$ from Theorems~\ref{thm:exactly-feasible-main} and~\ref{thm:approximately-feasible-main} without changing their families or accuracy parameters.

\begin{corollary}[Linear lower bounds at the original scales]
\label{cor:direct-original-scales}
For the families and accuracy parameters in Theorems~\ref{thm:exactly-feasible-main} and~\ref{thm:approximately-feasible-main}, every optimizer requires
\[
 T\ge\frac n{500}
\]
membership queries for all sufficiently large $n$. The feasibility guarantees, geometric parameters, and membership-oracle models are unchanged.
\end{corollary}

\begin{proof}
For the exactly feasible family, $s=n$, $a_d=2\overline g_n=\Theta(n^{-2})$, and $d=\Theta(n)$. The right side of \eqref{eq:direct-exact-accuracy} is therefore $\Theta(n^{3/2})$, while the accuracy $\eps_n$ in Theorem~\ref{thm:exactly-feasible-main} is $\Theta(n^{-2})$. Hence \eqref{eq:direct-exact-accuracy} holds for all sufficiently large $n$.

For the approximately feasible family, $s=n^2$ and $R_s/r_s=\Theta(n)$. The right side of \eqref{eq:direct-approximate-accuracy} is again $\Theta(n^{3/2})$, while the tolerance $\delta_n$ in Theorem~\ref{thm:approximately-feasible-main} is $\Theta(n^{-2})$. Thus \eqref{eq:direct-approximate-accuracy} also holds for all sufficiently large $n$. Theorem~\ref{thm:direct-membership-lower-bound} now gives $T\ge n/500$ in both cases.
\end{proof}

\subsection{Strongly convex quadratics}
\label{sec:quadratic-inverse-iteration}

The first version of this paper proved a weaker form of Theorem~\ref{thm:direct-gradient-lower-bound}, with a factor $\log q\log\log q$ and a smaller accuracy, by the inverse-iteration argument of this section. It uses the quadratics $f_{Q,b}$ of Section~\ref{sec:quadratic-family}. The constant $\rho$ is defined in \eqref{eq:inverse-constants}.

\begin{theorem}[Gradient lower bound for strongly convex quadratics, first version]
\label{thm:strongly-convex-gradient-lower-bound-v1}
There are universal constants $C,q_0>0$ such that the following holds for every $n\ge1$ and $\kappa\ge1$. Put
\[
 q:=\min\{\sqrt\kappa,n\}.
\]
There is an explicit family $\mathscr F_{n,\kappa}'$ of quadratic functions on $\R^n$, each $1$-strongly convex and $\kappa$-smooth, such that any quantum algorithm which, for every $f\in\mathscr F_{n,\kappa}'$, returns $\widehat x$ satisfying
\begin{equation}
 f(\widehat x)\le\min_{x\in\R^n}f(x)+\frac{\rho^2}{2}
 \label{eq:quadratic-constant-accuracy-v1}
\end{equation}
with probability at least $2/3$ makes at least
\begin{equation}
 C\frac{q}{\log(2+q)\,\log\log(3+q)}
 \label{eq:quadratic-gradient-lower-bound-v1}
\end{equation}
queries to the gradient phase oracle \eqref{eq:gradient-phase-oracle} in the worst case.
\end{theorem}

\begin{proof}
Fix a sufficiently large universal constant $q_0\ge200$. We first dispose of the range $q<q_0$. Let $e_1$ be the first standard basis vector and consider
\[
 f_\pm(x):=\frac12\norm{x\mp2\rho e_1}^2.
\]
Set $\mathscr F_{n,\kappa}':=\{f_+,f_-\}$. Both functions are $1$-strongly convex and $1$-smooth, hence $\kappa$-smooth. The success condition with accuracy $\rho^2/2$ requires the output to lie within distance $\rho$ of the corresponding minimizer. These two success regions are disjoint. A zero-query algorithm has the same output distribution for both functions and therefore cannot succeed with probability at least $2/3$ on each. Thus at least one query is necessary. Since $q<q_0$, this proves \eqref{eq:quadratic-gradient-lower-bound-v1} in this range after choosing the universal constant $C$ sufficiently small.

It remains to consider $q\ge q_0$. We use the functions $f_{Q,b}$ of Section~\ref{sec:quadratic-family} and set $\mathscr F_{n,\kappa}':=\{f_{Q,b}:Q\in\mathsf O(d),\ b\in\mathbb S^{d-1}\}$. By \eqref{eq:quadratic-inverse-accuracy} with $\eta=\rho$, on the success event $\norm{\widehat x-x_{Q,b}^*}\le\rho$, and projection onto the first $d$ coordinates produces a vector within distance $\rho$ of $A_Q^{-1}b$.

Now suppose that an optimizer satisfying the theorem uses at most $T$ gradient queries. Apply it repeatedly to the functions $f_{Q,b}$, choosing $b$ adaptively as in the inverse iteration of Section~\ref{sec:inverse-iteration}. By \eqref{eq:quadratic-inverse-accuracy}, every successful invocation supplies the same approximate inverse application used in Corollary~\ref{cor:one-inverse}. The majority-ball amplification, Rademacher initialization, and noisy inverse iteration are therefore unchanged.

The final test uses one more amplified batch of optimizer invocations with objective vector $v_k$. A successful output $x$ satisfies
\[
 \abs{v_k^Tx'-v_k^TA_Q^{-1}v_k}\le\rho,
\]
where $x'$ denotes its first $d$ coordinates. On $\mathcal E_{d,+}$,
\[
 v_k^TA_Q^{-1}v_k\le\frac1{3g_d},
\]
whereas on $\mathcal E_{d,-}$, after successful inverse iteration,
\[
 v_k^TA_Q^{-1}v_k\ge\frac9{20g_d}.
\]
For all sufficiently large $d$, the fixed error $\rho$ is smaller than both gaps to the threshold $2/(5g_d)$. The median test and the probability calculation in Proposition~\ref{prop:optimizer-determinant} therefore give a determinant predictor with average success probability greater than $1/2$ using
\[
 W=O(\log d\,\log\log d)
\]
optimizer invocations.

Each optimizer invocation uses at most $T$ gradient queries, and each gradient query is simulated with two matrix phase queries to $Q$. The determinant predictor therefore uses at most $2WT$ matrix phase queries. Theorem~\ref{thm:mv-determinant} in dimension $d$ gives
\[
 2WT\ge\frac d2.
\]
Hence
\[
 T=\Omega\!\left(\frac{d}{\log d\,\log\log d}\right).
\]
Combining this with \eqref{eq:quadratic-hard-dimension-order} proves \eqref{eq:quadratic-gradient-lower-bound-v1} after adjusting universal constants.
\end{proof}

\section*{Discussion}
We comment briefly on some outstanding open questions. The first natural question is whether the Fourier-rank polynomial method has applications to lower bounds beyond those considered here. Generalizations of the method that correspond to the threshold degree and approximate degree variants of the classic polynomial method are also of interest.

A second question concerns the dependence of the optimization lower bounds on the accuracy. To compare our results with the quantum lower bound of Garg, Kothari, Netrapalli, and Sherif~\cite{garg2020no}, we use their normalization. They minimize a $G$-Lipschitz convex function over a ball of radius $R$ and measure accuracy by the dimensionless ratio $\eps/(GR)$. In our setting, the objective $x\mapsto-c^Tx$ with $\norm c=1$ is $1$-Lipschitz and the feasible body lies in $B_2(0,R)$, so the corresponding ratio is $\eps/R$.

In this normalization, Corollary~\ref{cor:direct-normalized} shows that $\Omega(n)$ membership queries are necessary whenever $\eps/R\le1/(200\sqrt n)$. The quantum algorithms of~\cite{chakrabarti2020quantum,van2020convex} use $\widetilde O(n)$ membership queries, so the membership-query complexity is $\widetilde\Theta(n)$ in this regime. Garg, Kothari, Netrapalli, and Sherif prove that $\Omega((GR/\eps)^2)$ quantum queries to function values and subgradients are necessary when the dimension is at least $\widetilde O((GR/\eps)^4)$, that is, when $\eps/(GR)$ is at least of order $n^{-1/4}$ up to logarithmic factors. Projected subgradient descent matches this bound with $O((GR/\eps)^2)$ queries.

The two results use different models. Our bounds concern membership queries for linear optimization over ellipsoids whose inner radius is $\Theta(R/n)$, while their bounds concern first-order queries for Lipschitz functions on a ball. Neither result applies when $\eps/R$ lies between order $n^{-1/2}$ and order $n^{-1/4}$. At $\eps/R=\Theta(n^{-1/2})$, the quantities $n$ and $(R/\eps)^2$ are of the same order. It is therefore natural to ask whether the quantum query complexity is $\widetilde\Theta(\min\{n,(R/\eps)^2\})$ throughout the intermediate range. We leave this question open.

\section*{AI Usage Disclosure}
\phantomsection\label{sec:ai_use}
The ideas behind these proofs were developed in collaboration with Large Language Models. During the course of this project, we tried several models including Claude Opus 4.7 and Opus 4.8, as well as GPT-5.5 and GPT-5.6 Sol. The essential notion of Fourier rank was identified by each of these models but the attempts to form a lower bound proof based on these considerations turned out to be incorrect due to various scaling issues encountered while quantizing classical matrix-vector product lower bounds for eigenvalue estimation. GPT-5.6 Sol suggested the route via a determinant lower bound that allowed for these issues to be sidestepped, and led to the complete proofs presented here. Through a series of interactions with the model, GPT-5.6 Sol produced an essentially complete proof of the convex optimization lower bound that has been polished, verified, and substantially rewritten by the authors for clarity. The primary AI-contributed ideas in the reduction from determinant to convex optimization concern the bridge between appropriate binary membership queries and the continuous phase oracle model of the determinant lower bound. This covers all AI contributions to version 1 of the paper. The improvements in the proof that are presented in this version, are built upon a sketch of ideas from the authors, that were expanded into concrete proofs using GPT-6 Sol, and integrated into the document with help from GPT-6 Sol and Claude Opus 5.5.

The authors have manually verified each of the technical claims made by the LLM and are responsible for the correctness of all results in the paper.

\section*{Acknowledgment}
We thank Andrew Childs for informing us of his concurrent result, and for helpful discussions. We thank Rob Otter and Ruslan Shaydulin for executive support and valuable feedback on this project. We also acknowledge our colleagues at the Global Technology Applied Research Center of JPMorganChase.

\section*{Disclaimer}
This paper was prepared for informational purposes by the Global Technology Applied Research center of JPMorgan Chase \& Co. This paper is not a product of the Research Department of JPMorgan Chase \& Co. or its affiliates. Neither JPMorgan Chase \& Co. nor any of its affiliates makes any explicit or implied representation or warranty and none of them accept any liability in connection with this paper, including, without limitation, with respect to the completeness, accuracy, or reliability of the information contained herein and the potential legal, compliance, tax, or accounting effects thereof. This document is not intended as investment research or investment advice, or as a recommendation, offer, or solicitation for the purchase or sale of any security, financial instrument, financial product or service, or to be used in any way for evaluating the merits of participating in any transaction.

\bibliographystyle{alpha}
\bibliography{bibo}

\appendix
\section{Continuous exact membership queries}
\label{app:continuous-membership}

This appendix proves Corollary~\ref{cor:continuous-membership-determinant} by formulating an uncountable query register directly in $L^2$. Throughout, $n\ge2$, $(Y,\nu)$ is sigma-finite, $\mathcal K$ is separable, and $q:Y\to\R^n$ and $\tau:Y\to\R$ are measurable. Define
\begin{equation}
 \phi(Q,y):=(-1)^{\1[q(y)^TQq(y)\le\tau(y)]}.
 \label{eq:continuous-membership-phase}
\end{equation}
The map $(Q,y)\mapsto\phi(Q,y)$ is measurable and has modulus one.

\begin{lemma}[Continuous membership has rank-one growth]
\label{lem:continuous-membership-growth}
Let $M$ be the multiplication operator
\[
 (M\Psi)(Q,y):=\phi(Q,y)\Psi(Q,y).
\]
Then, for every $r\ge0$,
\[
 M\mathcal W_r(Y,\mathcal K)
 \subseteq\mathcal W_{r+1}(Y,\mathcal K).
\]
\end{lemma}

\begin{proof}
Put $A(y):=\norm{q(y)}^2$. If $A(y)=0$, the phase is independent of $Q$. The same is true when $A(y)>0$ and either $\tau(y)<-A(y)$ or $\tau(y)\ge A(y)$. Multiplication by these phases preserves every rank space.

It remains to consider
\[
 Y^\circ:=\{y:A(y)>0,\ -A(y)<\tau(y)<A(y)\}.
\]
For $y\in Y^\circ$, define
\[
 c(y):=\frac{\tau(y)}{A(y)},
 \qquad
 u(Q,y):=\frac{q(y)^TQq(y)}{A(y)}.
\]
Both maps are measurable and $u(Q,y)\in[-1,1]$.

For $c\in(-1,1)$, let
\[
 h_c(u):=(-1)^{\1[u\le c]},
 \qquad -1\le u\le1,
\]
and extend $h_c$ with period two. Its Fourier coefficients and $N$th Fej\'er mean are
\begin{align}
 \widehat h_c(k)
 &:=\frac12\int_{-1}^1 h_c(u)e^{-i\pi ku}\,du,
 \\
 F_N(c,u)
 &:=\sum_{k=-N}^N
 \left(1-\frac{|k|}{N+1}\right)
 \widehat h_c(k)e^{i\pi ku}.
 \label{eq:joint-Fejer}
\end{align}
For each $k$, the map $c\mapsto\widehat h_c(k)$ is continuous. Hence $F_N$ is jointly measurable. Positivity and unit mass of the Fej\'er kernel give
\[
 |F_N(c,u)|\le1.
\]
At every continuity point of the periodic function $h_c$, Fej\'er's theorem gives~\cite{katznelson2004introduction}
\[
 F_N(c,u)\longrightarrow h_c(u).
\]

Define the measurable global approximant $p_N$ by
\[
 p_N(Q,y):=
 \begin{cases}
 F_N(c(y),u(Q,y)),&y\in Y^\circ,\\
 \phi(Q,y),&\text{when the phase is constant in $Q$},\\
 1,&A(y)>0\text{ and }\tau(y)=-A(y).
 \end{cases}
\]
For $y\in Y^\circ$, expansion \eqref{eq:joint-Fejer} gives
\begin{equation}
 p_N(Q,y)
 =\sum_{k=-N}^N c_{N,k}(y)\chi_{\Gamma_k(y)}(Q),
 \qquad
 \Gamma_k(y):=\frac{\pi k}{A(y)}q(y)q(y)^T,
 \label{eq:continuous-Fejer-expansion}
\end{equation}
where the coefficients $c_{N,k}$ are measurable and bounded. Every $\Gamma_k$ is measurable and has rank at most one. Lemma~\ref{lem:controlled-phase-growth} shows that multiplication by each $\chi_{\Gamma_k(y)}$ maps $\mathcal W_r$ into $\mathcal W_{r+1}$. Multiplication by $c_{N,k}(y)$ acts only on the query register and preserves that space. The sum in \eqref{eq:continuous-Fejer-expansion} is finite, so
\begin{equation}
 p_N\Psi\in\mathcal W_{r+1}(Y,\mathcal K)
 \qquad
 \text{for every }\Psi\in\mathcal W_r(Y,\mathcal K).
 \label{eq:continuous-Fejer-rank}
\end{equation}

We next prove convergence in the product $L^2$ space. Fix $y$ with $A(y)>0$. Under Haar measure, $Qq(y)/\norm{q(y)}$ is uniform on the unit sphere, and $u(Q,y)$ has the distribution of one coordinate of a uniform spherical point. For $n\ge2$, this distribution is atomless. The threshold set $u(Q,y)=c(y)$ and the endpoint sets $u(Q,y)=\pm1$ therefore have Haar measure zero. The same conclusion holds for the endpoint case $\tau(y)=-A(y)$.

These exceptional subsets of $\mathsf{O}(n)\times Y$ are measurable. Tonelli's theorem turns the fiberwise Haar-null statements into one $(\mu_n\otimes\nu)$-null exceptional set. Consequently,
\[
 p_N(Q,y)\longrightarrow\phi(Q,y)
 \qquad
 \text{for almost every }(Q,y).
\]
For $\Psi\in\mathcal W_r(Y,\mathcal K)$, the bound $|p_N|\le1$ gives
\[
 |p_N-\phi|^2\norm{\Psi}^2
 \le4\norm{\Psi}^2.
\]
The right-hand side is integrable on $\mathsf{O}(n)\times Y$. Dominated convergence therefore yields
\begin{equation}
 \norm{(p_N-\phi)\Psi}_{L^2(\mathsf{O}(n)\times Y)}
 \longrightarrow0.
 \label{eq:continuous-membership-strong-limit}
\end{equation}
Each $p_N\Psi$ lies in the closed space $\mathcal W_{r+1}$ by \eqref{eq:continuous-Fejer-rank}. Its limit $\phi\Psi=M\Psi$ lies in the same space, proving the lemma.
\end{proof}

\begin{remark}[Analytic approximants]
The Fej\'er sums $p_N$ establish the strong limit \eqref{eq:continuous-membership-strong-limit}. They are not asserted to be unitary or physically implementable queries.
\end{remark}

We now return to the exact coherent membership oracle. For fixed $Q$, it is the measurable controlled permutation on $L^2(Y,\nu)\otimes\C^2$ corresponding to
\[
 U_Q\ket{y}\ket b
 =\ket y\ket{b\oplus\1[q(y)^TQq(y)\le\tau(y)]}.
\]
This ket notation is only mnemonic; the oracle is defined as an operator on square-integrable wavefunctions. In the Hadamard basis
\[
 \ket{\widehat z}
 =2^{-1/2}\sum_{b=0}^1(-1)^{zb}\ket b,
\]
the $z=0$ component is fixed and the $z=1$ component is multiplied by $\phi(Q,y)$. Thus one exact membership query and its inverse map rank $r$ amplitudes to rank at most $r+1$ by Lemma~\ref{lem:continuous-membership-growth}. Theorem~\ref{thm:Fourier-rank-method} places every $T$-query outcome probability in $\mathcal A_{2T}^{(n)}$. If $2T<n$, equation~\eqref{eq:reflection-L1} makes its determinant correlation zero, and the decision conclusion of Theorem~\ref{thm:Fourier-rank-method} gives average success probability $1/2$. This proves Corollary~\ref{cor:continuous-membership-determinant}.

\begin{remark}[Boundary values]
For a continuous $L^2(Y,\nu)$ query register, changing membership on a $\nu$-null subset of query points does not change the oracle as an operator. The argument above does not rely on that fact for a finite or countable register: for each fixed query point, equality in the threshold is Haar-null as a function of $Q$. The same non-strict boundary convention therefore works in all three settings.
\end{remark}

\end{document}